\documentclass[a4paper,twocolumn,11pt,accepted=2025-05-30]{quantumarticle}
\pdfoutput=1
\usepackage[numbers,sort&compress]{natbib}
\usepackage[utf8]{inputenc}
\usepackage[T1]{fontenc}
\usepackage[polish,spanish,english]{babel}
\usepackage{graphicx,amsmath,amsfonts,enumerate,amsthm,amssymb,mathtools,enumitem,thmtools,hyperref,mathdots,enumitem,centernot,bm,soul,bbm}
\usepackage{slashbox}
\usepackage[capitalise, noabbrev]{cleveref}
\usepackage{float}
\usepackage{physics}
\usepackage{colortbl}
\usepackage[table]{xcolor}
\usepackage{mathrsfs} 
\usepackage{lipsum}
\usepackage{tikz}
\usepackage{microtype}
\newcommand{\tikzxmark}{%
\tikz[scale=0.23] {
 \draw[line width=0.7,line cap=round] (0,0) to [bend left=6] (1,1);
 \draw[line width=0.7,line cap=round] (0.2,0.95) to [bend right=3] (0.8,0.05);
}}
\usepackage{comment}
\usepackage{pifont}
\usepackage[textsize=small, disable]{todonotes}
\setuptodonotes{inline, color=blue!20}

\graphicspath{{Figures/}}
\hypersetup{colorlinks=true,linkcolor=blue,citecolor=blue,urlcolor=blue}
\definecolor{mygreen}{rgb}{0.01, 0.31, 0.59}
\definecolor{myblue}{rgb}{0.01, 0.31, 0.59}
\definecolor{bluecyan}{rgb}{0.27, 0.66, 0.88}
\hypersetup{
 colorlinks=true, 
 linkcolor=myblue, 
 citecolor=mygreen, 
 urlcolor =myblue 
}

\def\>{\rangle}
\def\<{\langle}

\renewcommand{\v}[1]{\ensuremath{\boldsymbol #1}}

\definecolor{ppblue}{RGB}{46,117,182}
\definecolor{ppred}{RGB}{197, 90, 17}

\newcommand{\bl}[1]{{\color{ppblue}#1}}

\newcommand{\vga}[1]{ \textcolor{blue}{({\tt VGA}: #1)}}
\newcommand{\dg}[1]{ \textcolor{orange}{({\tt DG}: #1)}}
\newcommand{\jcz}[1]{{\color{green!66!black}({\tt JCz}: #1)}}

\theoremstyle{plain}
\newtheorem{thm}{Theorem}
\newtheorem{lem}[thm]{Lemma}
\newtheorem{prop}[thm]{Proposition}
\newtheorem{cor}[thm]{Corollary}
\newtheorem{obs}[thm]{Observation}
\newtheorem{conj}[thm]{Conjecture}
\theoremstyle{definition}
\newtheorem{defn}{Definition}
\newtheorem{const}{Construction}

\usepackage[normalem]{ulem}

\begin{document}

\title{Cyclic measurements and simplified quantum state\newline tomography}



\author{Victor Gonz{\'a}lez Avella}
\email{victor.gonzalez.avella@ua.cl}
\affiliation{Departamento de F{\'i}sica, Facultad de Ciencias B{\'a}sicas,
Universidad de Antofagasta, Casilla 170, Antofagasta, Chile}
\orcid{0000-0003-2633-6146}

\author{Jakub Czartowski}
\email{jakub.czartowski@ntu.edu.sg}
\affiliation{Doctoral School of Exact and Natural Sciences, Jagiellonian University, ul. Lojasiewicza 11, 30-348 Kraków, Poland}
\affiliation{Faculty of Physics, Astronomy and Applied Computer Science, Jagiellonian University, 30-348 Kraków, Poland}
\affiliation{School of Physical and Mathematical Sciences, Nanyang Technological University,
21 Nanyang Link, 637371 Singapore, Republic of Singapore}
\orcid{0000-0003-4062-833X}

\author{Dardo Goyeneche}
\email{dardo.goyeneche@uc.cl}
\affiliation{Departamento de F{\'i}sica, Facultad de Ciencias B{\'a}sicas,
Universidad de Antofagasta, Casilla 170, Antofagasta, Chile}
\affiliation{Instituto de F\'isica, Pontificia Universidad Cat\'olica de Chile, Casilla 306, Santiago, Chile}
\orcid{0000-0002-9865-4226}

\author{Karol Życzkowski}
\email{karol.zyczkowski@uj.edu.pl}
\affiliation{Faculty of Physics, Astronomy and Applied Computer Science, Jagiellonian University, 30-348 Kraków, Poland}
\affiliation{Center for Theoretical Physics, Polish Academy of Sciences, ul Lotników 32/46, 02-668 Warszawa, Poland}
\orcid{0000-0002-0653-3639}

\date{May 26, 2025}


\maketitle

\begin{abstract}
Tomographic reconstruction of quantum states plays a fundamental role in benchmarking quantum systems and accessing information encoded in quantum-mechanical systems. 
Among the informationally complete sets of quantum measurements, the tight ones provide a linear reconstruction formula and minimize the propagation of statistical errors. However, implementing tight measurements in the lab is challenging due to the high number of required measurement projections, involving a series of experimental setup preparations. In this work, we introduce the notion of cyclic tight measurements, which allow us to perform full quantum state tomography while considering only repeated application of a single unitary-based quantum device during the measurement stage. This type of measurement significantly simplifies the complexity of the experimental setup required to retrieve the quantum state of a physical system. 
Additionally, we design a feasible setup preparation procedure that produces well-approximated cyclic tight measurements in every finite dimension.

\end{abstract}

\section{Introduction}
Informationally complete quantum measurements play a crucial role in quantum information theory. They provide a physically admissible way to acquire full information concerning a state of a quantum system~\cite{nielsen2002quantum}. In particular, tight informationally complete quantum measurements~\cite{scott2006tight} provide a linear formula to reconstruct any quantum state. However, a common problem in tight measurements is that they cannot be efficiently implemented in the laboratory, in the sense that the amount of physical resources required to realize them grows exponentially with the number of parties. In quantum computing, the number of circuits required to implement a tight measurement for $n$-qubit systems typically grows at least as $2^{2n}$. This number arises from the fact that a $d$-dimensional Hilbert space, denoted as $\mathcal{H}_d$, requires at least $d^2$ rank-one projectors~\cite{scott2006tight}, associated to a tight quantum measurement.

Nonetheless, for a certain class of tight measurements, the experimental setup at the measurement stage is much simpler. For instance, the so-called \emph{cyclic mutually unbiased bases}  
\cite{Chau_2005_unconditionally,gow2007generation,kern2010complete,seyfarth2011construction,seyfarth2014structure,appleby2009properties}, are maximal sets of mutually unbiased bases (MUB) generated through iterations of a single unitary transformation. This means that repeated application of a single quantum circuit is enough to reconstruct the memory state of a quantum computer. 
Moreover, any eigenvector of such unitary transformation is distinguished by the fact that it has the same probability distribution with respect to the set of MUB bases. Such states, so-called \emph{MUB-balanced}~\cite{amburg2014states}, define minimum uncertainty states~\cite{sussman2007discrete,sussman2007minimum} and they are closely related to symmetric informationally complete (SIC) quantum measurements~\cite{appleby2014symmetric} and random-access codes~\cite{casaccino2008extrema}.

Cyclic MUB are known to exist in dimension 2 and in all even prime-power dimension~\cite{kern2010complete}. The usefulness of this remarkable kind of measurements {has been shown for quantum key distribution~\cite{Chau_2005_unconditionally,seyfarth2019cyclic}.} However, beyond $N$ qubit systems, cyclic MUB remain elusive. In particular, as we will show later, cyclic MUB do not exist for a qutrit system, and they cannot be found after extensive numerical searches in dimension~5.
%
This lack of solutions is the main motivation to introduce an extension of cyclic MUB, called \textit{cyclic $t$-designs}, given by complex projective $t$-designs composed by a set of orthonormal bases generated through a repeated iteration of a single unitary transformation. 
While the concept of cyclic $t$-designs is mathematically well-defined, their exact implementation may not always be feasible due to experimental constraints. For this reason, we also introduce a method that allows us to create an approximate cyclic $t$-design by using random Hamiltonians. The most notable advantage of cyclic $t$-designs is that full quantum state tomography can be implemented with a minimal amount of experimental resources, namely repeated use of a single unitary transformation and a measurement apparatus; see Figure~\ref{fig:shelves}.
\begin{figure}[t]
 \centering
 \includegraphics[width=.9\columnwidth]{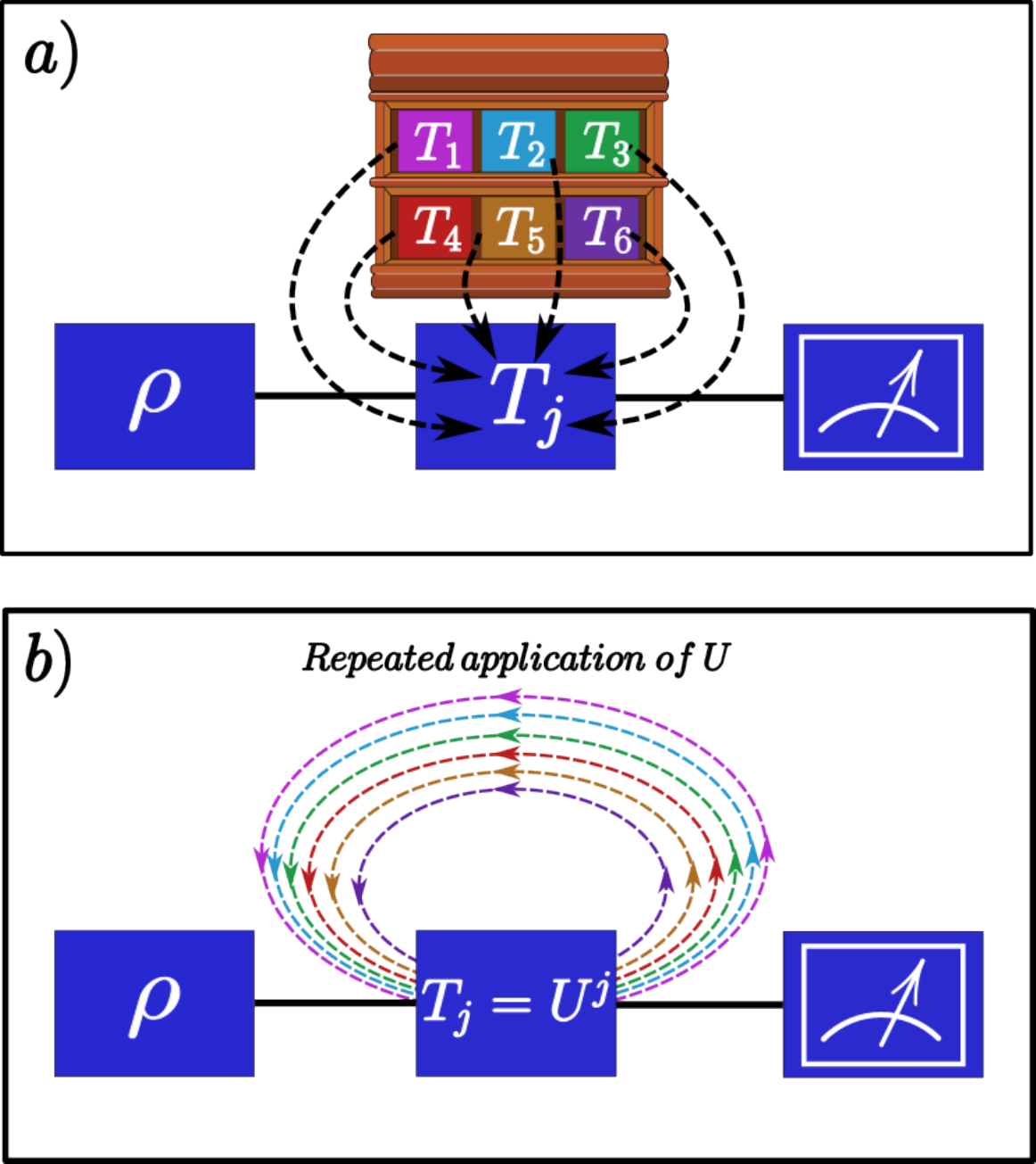}
 \caption{(Color online) Schematic representation of an experimental setup for quantum state tomography, that shows an unknown quantum system $\rho$, a black box operation $T_j$ and a final measurement with respect to the canonical basis. In part \emph{a)}, an experimental setup is required for the application of each $T_j$, whereas in \emph{b)} a single unitary transformation $U$ is applied to the system $j$ times, represented by extra dashed lines, thus producing $T_j = U^j$. {The main practical advantage of cyclic $t$-designs lies in their simplified experimental setup.}
 }
 \label{fig:shelves}
\end{figure} \medskip

This work is organized as follows. In Section~\ref{sec:setting}, we introduce the notions of cyclic MUB, cyclic $t$-designs and all the mathematical ingredients required to understand the work. In Section~\ref{sec:res}, we derive the main results of our paper, including the basic properties and construction of cyclic designs together with numerical investigation and the notion of approximate cyclic designs. We conclude the work with discussion in Section~\ref{sec:discussion}. Proof of the more complex results can be found in  Appendices \ref{app:proof_d2} to \ref{app:num}.






\section{Setting the Scene}\label{sec:setting}

The most general kind of measurements in quantum mechanics are given by positive operator valued measures (POVM), given by sets of positive semidefinite operators that sum up to the identity. Within these sets, the so-called \emph{informationally complete}, i.e., spannings the entire Hilbert space, are suitable to univocally reconstruct any quantum state. There are two essential properties we aim for when designing an experimental tomographic scheme: \textit{minimal propagation of statistical errors} and \textit{ease of experimental setup}. Minimizing error propagation leads us to the concept of tight informationally complete quantum measurements~\cite{scott2006tight}, while a simple experimental setup often suggests that only a limited set of natural measurements can be practically implemented.   


Along this line, compressed sensing techniques provide a way to reconstruct rank-$r$ quantum states with high probability from $O(rd\log^2 d)$ measurement settings that are locally applied~\cite{gross2010quantum}. Furthermore, any nearly pure quantum state can be reconstructed from the statistics of five measurement bases in any dimension~\cite{goyeneche2015five}. Also, $n$-qubit pure states can be reconstructed from $mn+1$ fully separable measurement bases, for any $m\geq2$, where $m$ can be increased to improve the fidelity of the reconstruction~\cite{pereira2022scalable}. 

On the other hand, the minimization of statistical errors propagation is satisfied by tight informationally complete quantum measurements, equivalent to the mathematical notion of \emph{complex projective $2$-designs}~\cite{hoggar1982t}. Interestingly, there is a lower bound for the average of entanglement in states that define tight quantum measurements, established for bipartite~\cite{wiesniak2011entanglement} and multipartite~\cite{czartowski2018entanglement} systems, implying that quantum entanglement is a fundamental resource for generating tight measurements. 

The aim of the present work consists in introducing a special class of tight quantum measurements, composed by sets of orthonormal bases, that are simple to implement in a laboratory, in the sense that all the measurement bases can be generated through powers of a single unitary transformation. From an experimental point of view, this implies that a single quantum device, iteratively applied before reaching a measurement apparatus, is sufficient to prepare the measurement stage, with the additional advantage of minimizing the propagation of statistical errors.\medskip

From now on, unless explicitly stated otherwise, we will consider greek indices going from $1$ to $d$, corresponding to the dimensionality of the underlying Hilbert space, and latin indices numbering the objects in questions, i.e. vectors or basis, thus going either from $1$ to $N$ or from $0$ to $k$.

 Let us start by recalling some basic definitions.

\begin{defn}[Mutually Unbiased Bases~\cite{ivonovic1981geometrical}]
Two orthonormal bases $\{\ket{\varphi_\alpha}\}_{\alpha=1}^{d}$ and $\{\ket{\psi_\beta}\}_{\beta=1}^{d}$ defined on a $d$-dimensional Hilbert space $\mathcal{H}_d$ are unbiased if $|\bra{\varphi_\alpha}\ket{\psi_\beta}|^2=\frac{1}{d}$, for all $\alpha,\beta=1,\dots,d$. A set of $m$ orthonormal bases are mutually unbiased (MUB) if they are pairwise unbiased.
\end{defn}
It is known that at most $d+1$ MUB exist in dimension $d$~\cite{ivonovic1981geometrical}. This upper bound is achieved in every prime~\cite{ivonovic1981geometrical} and prime power~\cite{wootters1989optimal} dimension $d$, whereas the question remains unknown in any other composite dimension, starting from $d=6$~\cite{brierley2009constructing}. There are several inequivalent constructions of maximal sets of MUB in prime power dimensions~\cite{wootters1989optimal, Bandyopadhyay2002, Chaturvedi2002a, Chaturvedi2002b, Klappenecker2003, Archer2005, Planat2005, Kibler2006} and a few constructions of small sets of MUB in other cases~\cite{grassl2004sic,bengtsson2007mutually,goyeneche2015mutually}. Further details about existence and construction of MUB can be found in a review published by \textit{Durt et al.}~\cite{durt2010mutually}. For our purposes, it is enough to restrict our attention to a cyclic procedure to generate maximal sets of MUB, defined as follows~\cite{seyfarth2011construction}.

\begin{figure}[t]
 \centering
 \includegraphics[width=.8\columnwidth]{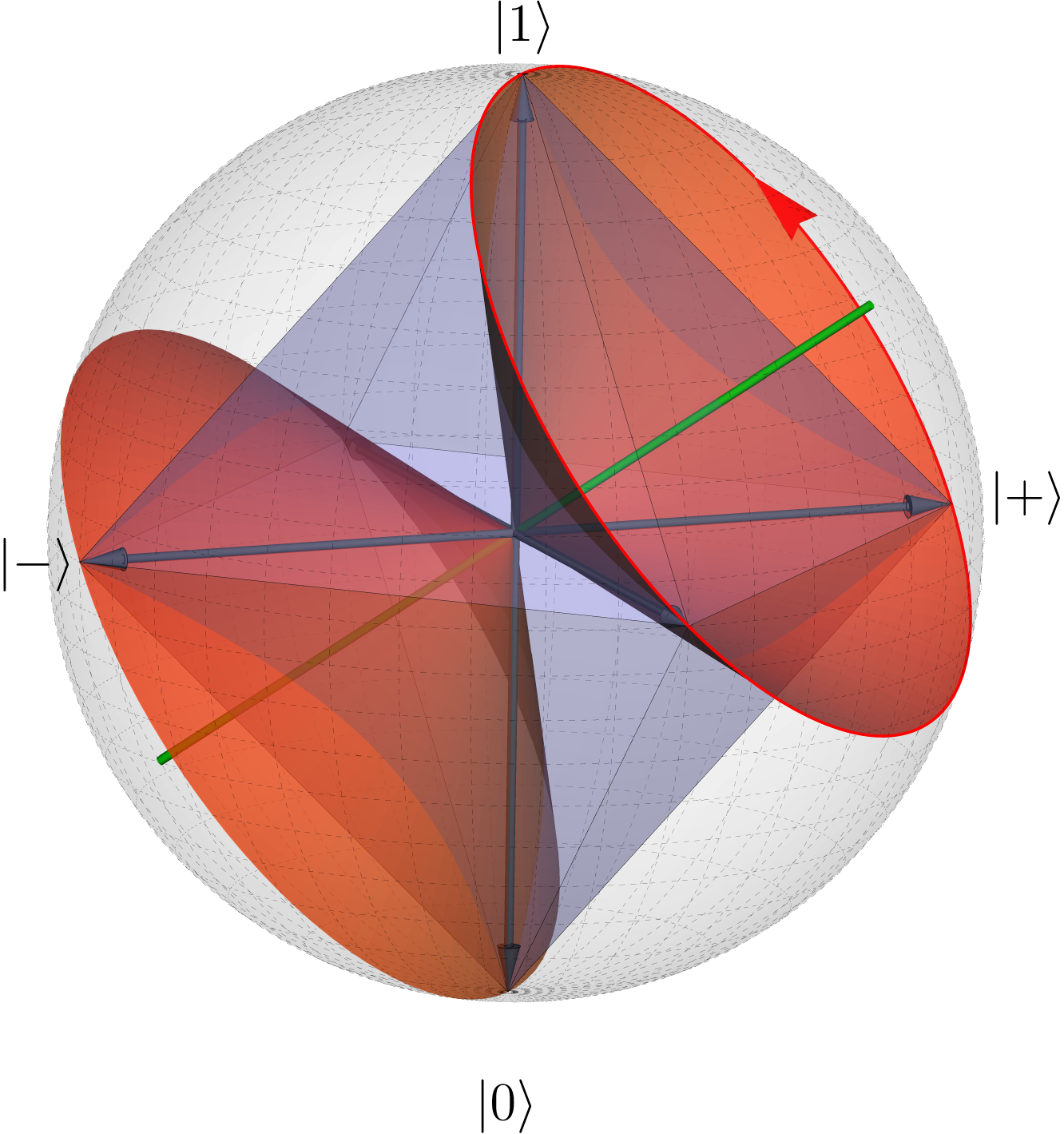}
 \caption{\emph{The cyclic MUB for a single qubit from Bloch ball geometry}: standard MUB consist of three bases, $\qty{\ket{0},\ket{1}}$, $\qty{\ket{-},\ket{+}}$ and $\qty{\ket{\odot},\ket{\otimes}}$, corresponding to pairs of vectors along $Z$, $X$ and $Y$ axes in the Bloch ball, respectively. The figure generated by the two orange cones looks similar a diabolo juggling prop. As all the states lie on a common cone (orange \textit{diabolo} shape), any of them can be transformed into any other by rotating around axis in the $(1,1,1)$ direction (green line) by angle equal to a multiple of $2\pi/3$. This corresponds to a unitary $U_1$ defined in~\eqref{U1}, with powers generating the MUB.
 }
 \label{fig:rot_MUB}
\end{figure}

\begin{defn}[Cyclic MUB]\label{cyclicMUB}
A unitary matrix $U$ of order $d$ generates MUB cyclically if the columns of the matrices \mbox{$U^0,U,U^2,\dots,U^d$} form a maximal set of $d+1$ MUB, where the upper index denotes matrix power.
\end{defn}
Known cyclic MUB define closed cycles, in the sense that $U^{d+1}=U^{0}=\mathbb{I}$.
The existence of cyclic MUB for a single qubit system can be easily visualized in the Bloch sphere. To this end, it is enough to restrict our attention to one Bloch vector of each basis, denoted as $\vec{r}_1$, $\vec{r}_2$ and $\vec{r}_3$. Cyclic MUB requires the existence of a single rotation $R$ in the real three-dimensional space such that $R\vec{r}_1=\vec{r}_2$, $R\vec{r}_2=\vec{r}_3$ and $R\vec{r}_3=\vec{r}_1$. This is simple to achieve when the rotation axis $\eta$ is chosen as $\vec{\eta}=\frac{1}{\sqrt{3}}(\vec{r}_1+\vec{r}_2+\vec{r}_3)$, see Figure~\ref{fig:rot_MUB}. In such case, the related unitary transformation in the complex Hilbert space is given by an enphased Hadamard matrix~\cite{seyfarth2011construction}:
\begin{equation}\label{U1}
U_1=\frac{1}{\sqrt{2}}\left(\begin{array}{rr}
1&-i\\
1&i
\end{array}\right).
\end{equation}

Remarkably, this solution provides a fundamental ingredient to construct cyclic MUB for $n'$ qubits when $n'=2^n$ with $n\in \mathbb{N}$, by following a simple recipe~\cite{seyfarth2019cyclic}:
\begin{const}\label{const_cMUB}
A maximal set of cyclic MUB for $n'=2^n$ qubits is generated by powers of
$U_{2^n}=2^{\frac{n-1}{2}}\,\operatorname{diag}[U_{2^{n-1}}]\qty(U_{2^{n-1}}\otimes U_{2^{n-1}})$, for any $n\in\mathbb{N}$, where $2^{\frac{n-1}{2}}\mathrm{diag}[U_{2^{n-1}}]$ is the diagonal unitary matrix of order $2^{n-1}$, whose main diagonal entries are defined by the concatenation all the rows of $U_{2^{n-1}}$, where $U_1$ is defined in (\ref{U1}). 
\end{const} 
The advantages of Construction~\ref{const_cMUB} are immediately recognized: \textit{(i)} It does not require the consideration of Galois fields or any other refined technique to generate a maximal set of MUB \textit{(ii)}~It suggests a simple experimental setup, as a single unitary is required to construct the full set of measurement bases. However, a construction method of such bases for any prime power dimension remains elusive, thus restricting its possible applications in quantum information theory. 

Following the example set down above, let us introduce a key notion for this work:

\begin{defn}[Complex projective $t$-design~\cite{hoggar1982t}]
Let $f_t(\ket{\psi})$ be a balanced polynomial function of order at most $t$ in both components of the state $\qty(\ket{\psi})_i$ from the space $\mathcal{H}_d$, and its conjugate $\qty(\bra{\psi})_i$. A set of pure states $\qty{\ket{\psi_i}\in\mathcal{H}_d}_{i=1}^m$ is called a \emph{complex projective $t$-design} if the average of any polynomial $f_t$ of degree at most $t$ over the set of states equals its average over the entire space,
 \begin{equation}
 \sum_{i=1}^m w_if_t(\ket{\psi_i}) = \int_{\mathcal{H}_d} f_t(\ket{\psi}) \dd{\psi},
 \end{equation}
 where the integral is taken over unitarily invariant
 measure $d\mu$ induced by the Haar measure on $\mathcal{U}(d)$ and all weights are set to $w_i = 1/m$. For any other set of weights $\qty{w_i}$, we call it a \emph{weighted complex projective $t$-design}.
\end{defn}
Equivalently, one can define a $t$-design in terms of averages of $t$-copy states,
\begin{equation}
 \frac{1}{N}\sum_{i=1}^N \op{\psi_i}^{\otimes t}= \int_{\mathcal{H}_d} \op{\psi}^{\otimes t} \dd{\psi},
\end{equation}
which is natural once we consider that the $t$-copy state $\op{\psi}^{\otimes t}$ contains all possible homogeneous monomials in the components of the state $\op{\psi}$ of degree $t$.

Furthermore, any set of $m$ vectors $\qty{\ket{\psi_i}}$ in dimension $d$ satisfies a family of inequalities, known as \textit{Welch bounds}~\cite{welch1974lower}:

\begin{equation}\label{WB}
 \frac{1}{N^2}\sum_{i,j=1}^N \abs{\ip{\psi_i}{\psi_j}}^{2t} \geq \frac{1}{\binom{d+t-1}{t}},
\end{equation}
for any $t\in\mathbb{N}$. The left hand side of \eqref{WB} is known as the \emph{frame potential}. Also, note that inequality \eqref{WB} is saturated if and only if the set $\qty{\ket{\psi_i}}$ defines a complex projective $t$-design~\cite{scott2006tight}. 

Complex projective $t$-designs for $t=2$ find applications in quantum state tomography~\cite{renes2004symmetric,lima2011experimental,bent2015experimental}, entanglement detection~\cite{chen2015general} and device-independent tests of quantum measurements~\cite{dall2017device}. Furthermore, its extension to unitary matrices, called unitary $t$-designs~\cite{dankert2009exact}, is a natural tool to implement quantum process tomography~\cite{scott2008optimizing}. The aforementioned full sets of $d+1$ MUB in dimension $d$ provide a canonic example of such measurements -- a fact that is easily proved using the Welch bound for $t=2$.

Similarly to complex projective $t$-designs, one can define $t$-designs for $d$-dimensional probability distributions comprising a $(d-1)$-dimensional simplex $\Delta_d$ with $d$ extreme points. 
\begin{defn}[Simplex $t$-design \cite{baladram2018explicit}]
A set of points $\vb{P} = \qty{\vb{p}_i\in\Delta_d}$ is called a $t$-design in the simplex if the average of any polynomial of order at most $t$ over $\vb{P}$ equals the average over the whole simplex when considering the flat Lebesgue measure
\begin{equation*}
 \ev{f}_{\Delta_d} \equiv \int_{\Delta_d} f(\vb{p}) \dd{p} = \frac{1}{\abs{\vb{P}}}\sum_{\vb{p}\in\vb{P}} f(\vb{p}) \equiv \ev{f}_{\vb{P}},
\end{equation*}
for every polynomial $f$ of degree at most $t$ in the entries of the probability distribution.
\end{defn}

The above two concepts are particular examples of a general notion of design, originally introduced as \emph{averaging sets} in~\cite{SEYMOUR1984213}.
%
This idea goes back all the way to Gaussian quadratures -- a notion used for numerical integration of continuous functions using a finite set of points~\cite{gauss1815methodus}.

In particular, in order to verify whether a given set $\vb{P}$ is a simplex design, it is enough to evaluate it over the monomial basis, e.g., for $2$-designs:
\begin{align*}
 \forall_{\alpha,\beta\in\qty{1,\hdots,d}} && \ev{p_\alpha}_{\Delta_d} & = \ev{p_\alpha}_{\vb{P}}, \\
 && \ev{p_\alpha^2}_{\Delta_d} & = \ev{p_\alpha^2}_{\vb{P}}, \\
 && \ev{p_\alpha p_\beta}_{\Delta_d} & = \ev{p_\alpha p_\beta}_{\vb{P}}.
\end{align*}
Indeed, from considering linear combinations, one can construct any polynomial of order $2$. A general formula for any monomial with arbitrary coefficients $\kappa_\alpha\in\mathbb{N}$ can be given in terms of generalised Beta function $\text{B}(x)$ \cite{baladram2018explicit, czartowski2025comment},
\begin{equation}\label{eq:general_simplex_av}
    \ev{\prod_{\alpha=1}^d p_\alpha^{\kappa_\alpha}}_{\Delta_d} \!\!\!\!= (d-1)!\,\text{B}\qty(\kappa_1+1,\hdots,\kappa_{d}+1).
\end{equation}
We are now in position to introduce a central notion of our work.
\begin{defn}[Cyclic measurements]
A \textit{cyclic measurement} is a collection of $k$ orthonormal bases generated through powers of a single unitary matrix, $U\in \mathcal{U}(d)$, and complemented by the computational basis. The aforementioned constellation of $(k+1) d$ vectors forming the cyclic measurements is given by the columns of the following matrices: $\mathbb{I},U,U^2,\dots,U^k$. 
\end{defn}

In particular, a cyclic measurement composed of $(k+1)d$ vectors that form a complex projective $t$-design is called a cyclic $t$-design. Here, it is simple to show that a necessary condition for the existence of a cyclic $2$-design is $k\geq d$, where the saturation of the inequality occurs for \emph{cyclic MUB}, introduced in Definition \ref{cyclicMUB}. Cyclic $t$-designs are useful in practice because their implementation requires the ability to prepare a single unitary transformation in the lab, whereas many unitary transformations are typically required in general. Furthermore, these measurements are informationally complete, in the sense that any quantum state can be reconstructed from the resulting statistical data. For these reasons, from now on we restrict our attention to cyclic $t$-designs. Here, we emphasize that constructions for $2$-designs composed by orthonormal bases, beyond MUB, already exist~\cite{mcconnell2007efficient,roy2007weighted,bodmann2016achieving,li2019efficient}. However, it is hard to check whether one of such designs admits a cyclic structure.\\

Let us start by noting an important consequence of cyclic designs. For a given set of quantum measurements, one can define uncertainty principles. The states minimizing uncertainty principles are called minimum-uncertainty~\cite{sussman2007discrete} or maximally certain states~\cite{oppenheim2010uncertainty}. These states play a well-known relevant role for the harmonic oscillator but also in Bell inequalities~\cite{oppenheim2010uncertainty}, SIC-POVM~\cite{appleby2014symmetric}, among others. In general, finding a state of minimal uncertainty is a difficult task.

The most general solution of cyclic $t$-designs for a qubit system, including cyclic MUB, is presented in Appendix~\ref{app:proof_d2}, and can be visualized as an inscribed regular prism or antiprism for even and odd values of $k+1$, respectively. Additionally, a general solution of cyclic $2$-designs in dimension 3 can be found in Appendix~\ref{app:cyclic_d3}. 
In general, the problem of constructing cyclic $t$-designs becomes apparently intractable, rendering it unlikely to be solved in its full generality.

As we show later, the following definition helps to link the notions of  complex projective $t$-designs with $t$-designs in the probability simplex. 
\begin{defn}[Decoherence of a quantum state]
 Consider a quantum pure state $\ket{\psi}\in\mathcal{H}_d$ and an orthonormal basis $W = \qty{\ket{\varphi_\alpha}}_{\alpha=1}^d$. One defines \emph{a decoherence of a state $\ket{\psi}$ with respect to the basis $W$} as a $d$-point classical probability distribution $
 \vb{p} = \qty{p_\alpha = \abs{\ip{\psi}{\varphi_\alpha}}^2}_{\alpha=1}^d$ \cite{czartowski2018entanglement}.
\end{defn}
For simplicity, throughout this work we consider decoherence with respect to the computational basis for all the cases. Note that this convention does not imply any restriction on the set of decohered states. On the other hand, note that decoherence of a pure state $\ket{\psi}$ is in fact equivalent to taking the main diagonal of a rotated state, $\vb{p} \equiv \operatorname{diag}\qty(W\op{\psi}W^\dagger)$. 

In what follows,
we will mostly restrict our attention to a special class of cyclic $2$-designs, which are closely related to two well-known mathematical tools: \textit{difference sets} and \textit{unistochastic matrices}. Let us start with the former~\cite{singer1938theorem}.

\begin{defn}[Difference set~\cite{hall1956survey}]\label{def:diff_set}
 Consider a $K$-element set $D$ of integers modulo $\nu$, $D\subset\mathbb{Z}_\nu$. 
 Such a set is called a $(\nu,\,K,\,\lambda)$-difference set if each element of a multiset  
 \begin{equation}
 \mathcal{D}_D = \qty{x - x'\mod{\nu}\mid x,x'\in D,x\neq x' },
 \end{equation}
 has multiplicity $\lambda$.
\end{defn}
%
In this work, due to their relevance for the topic of cyclic measurements, we focus on difference sets with $\lambda=1$. Some properties of this kind of difference sets can be found in~\cite{mann1952some}. Note that in such a case one finds the order of the modulo to be bounded from below, $\nu\geq K(K-1)+1$. In particular, sets saturating this bound are called \emph{perfect circular Golomb rulers} and the set of differences they generate is equal to $\mathcal{D}_D = \mathbb{Z}_{K(K-1)+1}\setminus\qty{0}$. For instance, the set $D=\{1,2,4\}$ is a $(7,3,1)$-difference set composed by $K=3$ elements, such that its differences, modulo $\nu=K(K-1)+1 = 7$, produce the subset of modulo subgroup of integers $\mathbb{Z}_7$, each of them occurring exactly $\lambda=1$ times, see Table~\ref{table1}.
\begin{table}[t]
\begin{center}
\begin{tabular}{ccr|c}
\multicolumn{3}{c|}{Differences}&$\mathrm{mod}\,7$\\
\hline
$1-2$&=&$-1$&$6$\\
$2-1$&=&$1$&${\color{red}\underline{\mathbf{1}}}$\\
$1-4$&=&$-3$&${\color{red}\underline{\mathbf{4}}}$\\
$4-1$&=&$3$&$3$\\
$2-4$&=&$-2$&$5$\\
$4-2$&=&$2$&${\color{red}\underline{\mathbf{2}}}$
\end{tabular}
\end{center}
\caption{Sets of $k$ integer numbers such that its differences modulo $v$ produce all integers in a set $D$, each of them occurring $\lambda$ times are called difference sets. Here, we show the difference set $D=\{1,2,4\}$, characterized by parameters $(v,k,\lambda)=(7,3,1)$. The suitable differences are highlighted in red (underlined in black copy).}
\label{table1}
\end{table}
Constructions of such sets for prime $K$ can be found in~\cite{singer1938theorem,ImreZ1993}. Furthermore, difference sets for any $K$ can be found based on Mian-Chowla sequence~\cite{Mian_Chowla_1944, Sloane_Plouffe_1995}, which is a self-generating set $\qty(a_1,a_2,\hdots)$ that is generated by a greedy algorithm that sets $a_n$ to the smallest integer such that all the differences $a_i - a_j$ are different for $i<j\leq n$. The first few terms of the sequence are given by
\begin{equation}
 (1,2,4,8,13,21,31,45,\hdots). \label{eq:mian_chowla}
\end{equation} 
An upper bound for each element of the sequence is given by $a_n \leq n^3/6 + O(n^2)$, where the approximation $a_n\approx n^3/\log^2(n)$ is conjectured. This sequence provides a difference set for $K = n$ and $\nu \geq 2a_n + 1$ for every $n$.
In general, difference sets are closely related to further combinatorial notions like Hadamard matrices, orthogonal arrays and linear codes, see the book of Hedayat \textit{et al.} for further details~\cite{hedayat1999orthogonal}. 

We will now proceed to recall two further notions relevant to the cyclic $t$-designs -- bi and unis\-tochastic matrices~\cite{rajchel2018robust}.
\begin{defn}[Bistochastic matrix]
  A matrix $B\in\mathbb{R}^{d\times d}$ is called \textit{bistochastic} (doubly stochastic) if
 \begin{equation}
 B_{\alpha\beta}\geq 0,\quad\sum_{\alpha=1}^d B_{\alpha\beta} = \sum_{\beta=1}^d B_{\alpha\beta} = 1.
 \end{equation}
\end{defn}
\begin{defn}[Unistochastic matrix]
 A bistochastic matrix $B\in\mathbb{R}^{d\times d}$ is called \textit{unistochastic} if there exists a unitary matrix $U$ of size $d$ such that
 \begin{equation}
 B_{\alpha\beta} = \abs{U_{\alpha\beta}}^2.
 \end{equation}
\end{defn}
For instance, the flat matrix $B_{\alpha\beta}=1/d$ is unistochastic for any $d\in\mathbb{N}$ due to the existence of the Fourier matrix $F_{\alpha\beta}=\frac{1}{\sqrt{d}}e^{2\pi i (\alpha-1)(\beta-1)/d}$. Unistochastic matrices play an important role in particle physics. For example, to find the unitary \textit{Cabibbo-Kobayashi–Maskawa} matrix 
\cite{cabibbo1963unitary,kobayashi1973cp} from its related unistochastic matrix is a challenging problem. Amplitudes of this matrix contain the complete information about weak decays that change the flavour, determined by the \textit{weak universality}~\cite{cabibbo1961electron}. Unistochastic matrices are also important for studying different mixtures of density matrices~\cite{bengtsson2003mix}, and for constructing equiangular tight frames~\cite{goyeneche2017equiangular}. It is known that any bistochastic matrix of order 2 is also unistochastic, whereas necessary and sufficient conditions for a bistochastic matrix of order 3 to be unistochastic are known~\cite{jarlskog1988unitarity}. For any higher order some necessary conditions are known; see e.g.~\cite{dictua2006separation}, but the full characterization remains open.

A simplification of this problem arises when considering circulant bistochastic matrices, i.e. when every row of the matrix is given by shifting to the right the previous row. However, even in such case the full problem is challenging.  A partial solution to this problem was recently found, which solves the circulant case when the bistochastic matrix has only two different entries~\cite{rajchel2018robust}.  Before showing this result, let us introduce some definitions.

A complex Hadamard matrix is a square matrix $H$ of order $d$ satisfying $HH^{\dag}=d\,\mathbb{I}$ and $\abs{H_{\alpha\beta}}=1$, where $H^{\dag}$ denotes the adjoint of $H$. 
 A matrix 
$H$ is called {\sl robust}
\cite{rajchel2018robust},
if $\Pi_2 H\Pi_2$ is a 2-dimensional complex Hadamard matrix,
 where $\Pi_2$ is any projection onto a 2-dimensional space spanned by two vectors of the computational basis, i.e. $\Pi^{\alpha\beta}_2=\op{\alpha}+\op{\beta}$, $\alpha\neq \beta$. This is equivalent to say that 
\begin{equation}
 \mqty(H_{\alpha\alpha} & H_{\alpha_\beta} \\ H_{\beta\alpha} & H_{\beta\beta}),
\end{equation}
is a complex Hadamard matrix, for any $\alpha\neq \beta$.

We are now in a position to state the following result.

\begin{lem} [\cite{rajchel2018robust}]
\label{lem:rob_had}
Let $B$ be a bistochastic circulant matrix of order $d$, defined by $B_{\alpha\beta}=a(1-\delta_{\alpha\beta})+b\,\delta_{\alpha\beta}$, where $\delta_{\alpha\beta}$ is the Kronecker $\delta$ function. If there exists a robust Hadamard matrix of order $d$ then $B$ is also unistochastic, for any $a,b\geq0$ such that $a^2+b^2=1$.
\end{lem}
\begin{proof}
It is simple to show that if $H$ is a robust complex Hadamard matrix of order $d$ then
$U=\sqrt{a}\mathfrak{D}+\sqrt{b}(H-\mathfrak{D})$ is unitary, where $\mathfrak{D}$ is a diagonal matrix such that $\mathfrak{D}_{\alpha\alpha}=H_{\alpha\alpha}$, provided that $a,b\geq0$ and $a^2+b^2=1$. Thus, $B_{\alpha\beta}=|U_{\alpha\beta}|^2$, and $B$ is unistochastic.
\end{proof}
An important geometrical interpretation of this result is that robust complex Hadamard matrices define rays composed entirely out of unistochastic matrices, within the larger space of bistochastic matrices, the so-called \textit{Birkhoff polytope}.
Here, note that matrix $B$ is circulant, whereas the underlying matrix $U$, defined in the proof of Lemma \ref{lem:rob_had}, is not necessarily circulant. Robust Hadamard matrices exist in infinitely many dimensions, and they are related to well-known classes of matrices such as symmetric conference matrices. There is a further relevant class of matrices, closely related to robust Hadamards. A Hadamard matrix $H$ is called \textit{skew} if $H+H^T=2\mathbb{I}$, where $T$ denotes transposition. It is known that any skew Hadamard matrix is robust, see Lemma 2.6 in \cite{rajchel2018robust}. Furthermore, any robust Hadamard matrix is sign equivalent to a skew Hadamard matrix, meaning that these matrices differ at most in sign changes applied either to rows or columns. A survey about the existence of skew Hadamard matrices can be found here~\cite{koukouvinos2008skew}.

 In Section \ref{sec:res}, we will use Lemma \ref{lem:rob_had} to show important results related to the existence of cyclic $t$-designs.

\section{Results}\label{sec:res}

In the following, we present our results on cyclic $t$-designs. Analytical findings are outlined in Section ~\ref{res1}, where we present some general properties and constructions based on simplex designs and difference sets. In Section~\ref{res2}, we show a method to approximate cyclic $t$-designs through the use of random Hamiltonians. This approach may be useful when dealing with experimental limitations. Here, we discuss how this method allows us to estimate a quantum state by a reconstruction formula. Our numerical findings are shown in Section~\ref{res3}, where we show a simple procedure to find cyclic designs in any dimension $d$ by using numerical optimization.  Some examples of numerical solutions in dimension $d=4$ are shown in Appendix \ref{app:num}.


\subsection{Basic properties of cyclic designs}\label{res1}

Let us start with a simple observation.
\begin{obs} \label{obs:nec_cond}
 Let $W = \qty{\ket{\psi_\alpha}\in\mathcal{H}_d}$ be a complex projective $t$-design. Then, for any orthonormal basis $B = \qty{\ket{b_\alpha}:\braket{b_\alpha}{b_\beta} = \delta_{\alpha\beta}}$, there is a $t$-design in the probability simplex $\Delta_d$, given by its decoherence with respect to the basis $B$, that is,$$\vb{P}_B = \qty{\vb{p}_\beta = \qty{\abs{\braket{\psi_\alpha}{b_\beta}}^2}_{\beta=1}^{d}}_{\alpha}.$$
\end{obs}

The above observation, in line with similar results in~\cite{czartowski2019isoentangled,iosue2024projective}, leads us to the following no-go property:
\begin{cor} \label{cor:nogoCPdes}
 Suppose that there exists a basis $\{\ket{b_\beta}\}$ for which decoherence of a set $\qty{\ket{\psi_\alpha}\in\mathcal{H}_d}$ is not a $t$-design in the probability simplex $\Delta_d$. Then, $\qty{\ket{\psi_\alpha}\in\mathcal{H}_d}$ is not a complex projective $t$-design.
\end{cor}
From now on, we will write $U=V\Lambda V^{\dag}$, where $\Lambda = \sum_{\alpha=1}^d \lambda_\alpha \op{\alpha}$  is a diagonal matrix containing the eigenvalues of $U$. Based on Corollary \ref{cor:nogoCPdes}, we can formulate the following result about cyclic designs.

\begin{thm}
 The set $\qty{U^i\ket{\beta},\,U\in \mathcal{U}(d)}_{i=0,\beta=1}^{k,d}$ is a cyclic $t$-design only if the set $\qty{\ket{v_\beta} = V^\dagger \ket{\beta}}_{\beta=1}^d$ provides, by decoherence, a $t$-design in the probability simplex $\Delta_d$.
\end{thm}
\begin{proof}
%
 Let $V$ be the unitary matrix that diagonalizes $U$, i.e., $U=V\Lambda V^\dagger$, and apply $V^\dagger$ to the entire 2-design, such that we will be consider $V^\dagger U^i = \Lambda^i V^\dagger$. 
Thus, the rotated cyclic 2-design is given by the set $\qty{\Lambda^i\ket{v_\beta},\,U\in \mathcal{U}(d)}_{i=0,\beta=1}^{k,d}$. Since $\Lambda$ is a diagonal unitary matrix, for a fixed index $\beta$, all vectors $\Lambda^i \ket{v_\beta}$ yield the same probability distribution by decoherence; thus, all $(k+1)d$ vectors in the design generate $d$ points with $(k+1)$-fold degeneracy in the probability simplex $\Delta_d$. By Corollary \ref{cor:nogoCPdes} 
, these $d$ points have to be a $t$-design in the simplex.
\end{proof}
\noindent Due to the above property, cyclic designs are very limited in terms of the degree $t$, as they require existence of $d$-point $t$-designs in $\Delta_d$. We demonstrate the resulting limitations below.

\begin{thm} \label{2d_3-design}
 Any cyclic $2$-design in dimension $d=2$ is also a cyclic $3$-design. Moreover, cyclic $4$-designs do not exist for $d=2$.
\end{thm}
\begin{proof}
The fact that a cyclic 2-design is also a cyclic 3-design in dimension $d=2$ is proven in Appendix~\ref{app:proof_d2}. Impossibility of generating a $4$-design in dimension $d\!\!=\!\!2$ is a direct consequence of Corollary~\ref{cor:nogoCPdes}. 
More precisely, consider a set of $N$ points in $\Delta_2$, which is equivalent to a set of numbers $\qty{0\leq x_i \leq 1}_{i=1}^N$. They need to satisfy a set of equalities for the averaged moments up to $t$ 
\begin{equation}\label{eq:2_d_moments}
    \frac{1}{N} \sum_{i=1}^N x_i^t = \frac{1}{t+1},
\end{equation}
with all the other moments linearly dependent. 
For $N = 2$, we find a pair of points $x_\pm = \frac{1}{2}\pm\frac{1}{\sqrt{12}}$ that satisfies equations (\ref{eq:2_d_moments}) for $t = 1, 2, 3$. However, a set of real solutions $\{x_i\}$ that satisfies these equations for all $t=1,2,3,4$ does not exist. 
Hence, there exist no two states in dimension $d=2$ that decohere to a simplex 4-design which, by Corollary \ref{cor:nogoCPdes}, completes the proof.
\end{proof}

Theorem \ref{2d_3-design} is directly generalizable to arbitrary dimension $d$, as we show below.
\begin{thm}\label{thm:deg_less_4}
    Cyclic $t$-designs do not exist for $t>3$, in any dimension $d$.
\end{thm}
\begin{proof}
    Consider a set of probability vectors $\vb{P} = \qty{\vb{p}^{(i)}}_{i=1}^N$. The matrix 
    $M^{\alpha\beta}_{\mu\nu} = \ev{p_\alpha p_\beta p_\mu p_\nu}_{\vb{P}}$, with $\alpha,\beta,\mu,\nu=1,\dots,d$,
    can be interpreted as a sum of rescaled projectors onto double-copy states $\vb{p}^{(i)\otimes 2}$ which is related to the fact that the matrix $M$ can be represented as $\sum_{i=1}^N \vb{p}^{(i)\otimes 4} = \sum_{i=1}^N \qty(\vb{p}^{(i)\otimes 2})^{\otimes 2}$, where $\vb{p}^{(i)\otimes 4}$ contains all monomials of order $4$ in the components of $\vb{p}^{(i)}$. If we require $\ev{p_\alpha p_\beta p_\mu p_\nu}_{\vb{P}} = \ev{p_\alpha p_\beta p_\mu p_\nu}_{\Delta_d}$ with righthand side given by eq. \eqref{eq:general_simplex_av}, we find that $\operatorname{rank}(M) = \binom{d+1}{2} > d$, which is shown in Appendix \ref{app:rank_M_matrix}. Furthermore, it is straightforward that $N \geq \operatorname{rank}(M)$. Thus, there are no simplex 4-designs composed of $N=d$ points, which would be necessary for a cyclic $4$-design. 
\end{proof}

In addition, let us put forward the following conjecture.
\begin{conj}
    There are no $d$-point simplex 3-designs in $\Delta_d$ for $d\geq3$.
\end{conj}

This conjecture is based on the fact that there is no evidence in the literature for existence of $d$-point simplex $3$-designs for $d\geq3$ 
. 
Even in the smallest case of $d=3$, contrary to the claims presented in \cite{baladram2018explicit}, rudimentary evaluation of a 3-point simplex 3-design put forward therein shows that the 3-point arrangement fails for averages of the form $\ev{p_\alpha^2 p_\beta}$ with $\alpha\neq \beta$; additionally, by considering a general form of 2-designs in $\Delta_3$, as given in \cite{czartowski2019isoentangled}, one can optimize over a single variable to show nonexistence of such structures -- proof is presented in Appendix~\ref{app:baladram_invalidation} 
. 

More generally, the following line of geometric reasoning can be put forward. Consider a set of $d$ vectors $\vb{q}^{(i)} \in \mathbb{R}^N$ 
for which $q_i^{(\alpha)} = p^{(i)}_\alpha$; this translates, roughly, to taking rows of a matrix as vectors, instead of columns. By considering requirements for simplex 1-design, we see that the vectors are all restricted to a plane $\sum_i q^{(\alpha)}_i = N/d$, making them effectively $(N-1)$-dimensional. The conditions on $\ev{p_\alpha^2}$ puts them onto a sphere of radius $R^2 = 2N/(d(d+1))$, thus fixing their freedom to $(N-2)$ parameters per vector, which can be related to the spherical coordinates on a hyperplane. Finally, the set of conditions $\ev{p_\alpha p_\beta}$ translates to equal angles between vectors, $\vb{q}^{(i)}\cdot\vb{q}^{(j)} = \cos\theta = \text{const.}$ for $i\neq j$. It is well known that this defines uniquely, up to rescaling and displacement, a $d$-point regular simplex embedded in $\mathbb{R}^N$,
 which due to confinement to an $(N-1)$-dimensional hyperplane has no more than $N$ points, thus showing that $N \geq d$ for 2-designs. In addition, such a simplex is restricted only to rotational degrees of freedom, giving exactly $\binom{N-1}{2}$ angles.
 Finally, the restrictions imposed by fixing $\ev{p_\alpha^3}$ restrict vectors further to $(N-3)$ free parameters per vector at most, and a total of $\binom{N-1}{2} - N$ free parameters for the entire set when taking into account the already fixed simplicial structure. The remaining conditions coming from $\ev{p_\alpha p_\beta^2}$ and $\ev{p_\alpha p_\beta p_\mu}$, which can be easily counted as $d(d-1) + \binom{d+2}{3}$, need to be satisfied simultaneously using the remaining freedoms. Explicit analytical proof, as presented in Appendix \ref{app:baladram_invalidation}, shows that they cannot be sa\-tisfied for $d = 3$. In addition, numerical experiments show that the best approximation of a $d$-point 3-design in $\Delta_d$ is found by setting the points to $\qty[a, (1-a)/(d-1), \hdots, (1-a)/(d-1)]$ with $a = \frac{d + \sqrt{d + 1} - 1}{d \sqrt{d + 1}}$ set to satisfy the 2-design condition exactly as a necessary condition for a 3-design. 

Finally, based on the evidence and our knowledge, it is reasonable to believe that the sequence of minimal numbers $N_*(d,t)$ of points in a simplex $t$-design as a function of dimension $d$ is strictly monotonic for all $t\geq 2$, including $t = 3$, thus leading to the conjecture put forward above. Note that it has been shown constructively shown that $N_*(d,2) \leq d$ and, by geometric discussion above, $N_*(d,2) \geq d$, thus leading to conclusion that $N_*(d,2) = d$, which is strictly monotonic. Additionally, as presented in Appendix \ref{app:rank_M_matrix}, $N_*(d,4)$ is lower-bounded by a strictly monotonic function $\binom{d+1}{2}$, with similar bounds conjectured for all $t \geq 2$. Additionally, we know that $N_*(3,3) > 3$. Thus, should strict monotonicity hold for $N_*(d,3)$, one would have $N_*(d,3)>d$ for all $d\geq 3$.
If proven true, it would lead to the following as a corollary.

\begin{conj} \label{always_2-design}
For $d \geq 3$, a cyclic $t$-design exists for $t=1$ or $t=2$.
\end{conj}




\bl{}

Independently from the above, one can furthermore demonstrate that not only the full set of cyclic MUB in $d=3$ does not exist, but not even a single complex Hadamard matrix can be a part of a cyclic $t$-design in such a dimension. The proof of this fact is presented in Appendix~\ref{App:noMUBd3}.\medskip

 By using the notion of difference set introduced in Definition \ref{def:diff_set}, below we introduce the main result of the work -- a construction of cyclic $2$-design in dimension $d$ given a difference set and a basis proceeding from a simplex $2$-design, proved in Appendix \ref{app:cd_anyD_proof}.

\begin{thm} \label{thm:cyclic_construction}
 Consider a difference set \mbox{$D\! =\! \qty{N_\beta}$} with parameters $(k+1,d,1)$ and a basis forming a unitary matrix $V^\dagger = \qty{\ket{v_\beta}}_{\beta=1}^{d
 }$ which yields by decoherence a 2-design in the probability simplex~$\Delta_d$. Let $\lambda_\beta\! =\! \exp(i \frac{2\pi}{k+1} N_{\beta})$ and $\Lambda = \sum_\alpha \lambda_\alpha \op{\alpha}$. Then, the set $\qty{V \Lambda^j V^\dag}_{j=0}^k$ is a cyclic $2$-design with $U = V \Lambda V^\dag$.
\end{thm}
\begin{obs}
 {\color{black}For the construction in Theorem \ref{thm:cyclic_construction} all the eigenvalues of $U$ are $(k+1)$-th roots of the unity and, thus, $U^{k+1} = \mathbb{I}$.}
\end{obs}
\noindent The above construction leads to the existence of cyclic $2$-designs for every $d$ where a basis $V^\dagger$ yielding simplex 2-design by decoherence exists, given sufficiently large $k$, and this is guaranteed for $k \geq 2a_d + 1$, where $a_d$ is the element of Mian-Chowla sequence~\eqref{eq:mian_chowla} \cite{Mian_Chowla_1944} and existence of bases yielding simplex 2-design via decoherence, which can be generated using rudimentary minimization methods, as implemented in a \textit{Mathematica} notebook available online~\cite{GitHubTest}, which we have found to work up to $d = 100$ due to computational power limitations. Some examples of matrices yielding a simplex $2$-design in the simplex are shown in Appendix \ref{app:decoherence_examples}. However, reverse process -- deciding whether a given bistochastic matrix has its unitary counterpart -- is an open problem for matrices of order \mbox{$d\geq 5$}~\cite{Auberson1991, bengtsson2005birkhoffs}. Furthermore, using Lemma \ref{lem:rob_had} and results from~\cite{rajchel2018robust, koukouvinos2008skew} we see that for infinitely many dimensions where robust Hadamard matrices exist, the underlying basis $V^\dagger$ can assume a particularly elegant form with just two amplitudes, $$\qty{a=\frac{-d+\sqrt{d+1}+1}{d \sqrt{d+1}},b = \frac{1}{d-\sqrt{d+1}+1}}.$$ However, due to the minimal size of the underlying $\mathbb{Z}_k$ for difference sets, one cannot achieve $k<d(d-1)$ by using the above construction. As a consequence we have the following two statements
\begin{thm} \label{thm:inf_dim}
    Cyclic 2-designs exist in an infinite family of dimensions.
\end{thm}
\begin{conj}[Numerical]\label{conj:all_dim}
    Cyclic 2-designs exist in an infinite family of dimensions.
\end{conj}
\begin{proof}
    The statement of Theorem \ref{thm:inf_dim} follows from existence of robust Hadamard matrices, as stated in Lemma \ref{lem:rob_had}, providing basis $V$, combined with Mian-Chowla sequence in eq. \eqref{eq:mian_chowla}, which provides an underlying difference set. Extension to Conjecture \ref{conj:all_dim} is made numerically by the means described above.
\end{proof}
Theorem \ref{thm:cyclic_construction} yields similar results as a construction of almost-minimal (weighted) 2-designs introduced in~\cite{iosue2024projective}. Nevertheless, the result just described cannot be generated by a single unitary. Moreover, the weighing involved in the aforementioned work necessarily implies considering additional sampling, increasing the cost of physical implementations.\medskip 



As a further comment, H. Zhu conjectured that any 2-design composed by at most $d(d+1)$ elements in dimension $d$ is either a SIC-POVM~\cite{renes2004symmetric} or a maximal set of MUB, see~\cite{zhu2015mutually}. We emphasize that a part of this conjecture is already resolved. That is, a set of $d+1$ orthonormal bases defines a 2-design if and only if the bases are MUB, see Theorem 3.3 in~\cite{roy2007weighted}. Below, we provide a much simpler proof of this fact based on geometrical properties of MUB in the Bloch hypersphere.
\begin{prop}\label{prop2d_d+1}
A $2$-design composed by $d+1$ orthonormal bases necessarily corresponds to $d+1$ MUB.
\end{prop}
 \begin{proof}
The proof starts by noting that any $d$-dimensional quantum state can be written as 
\begin{equation}\label{rhoGellMann}
\rho=\frac{1}{d}\left(\mathbb{I}+\sqrt{\frac{d(d-1)}{2}}\vec{r}\cdot\vec{\sigma}\right),
\end{equation}
where $\vec{r}\in\mathbb{R}^{d^2-1}$ is the Bloch vector associated to $\rho$ in the Bloch hypersphere, and $\vec{\sigma}$ is a vector of matrices with entries given by the generalized Gell-Mann matrices~\cite{georgi2000lie}, satisfying $\mathrm{Tr}(\sigma_j\sigma_k)=2\delta_{jk}$. Let $\vec{r}^{\,i}_j$ be the $i$-th Bloch vector associated to the $j$-th MUB basis. Thus, from combining (\ref{WB}) and (\ref{rhoGellMann}), the $2$-design condition in the Bloch hypersphere reduces to: 
\begin{equation}\label{WB2}
 \sum_{i\neq i'}\sum_{j,j'=1}^{d}\left(\frac{1}{d}+\frac{d-1}{d}(\vec{r}^{\,i}_j\cdot\vec{r}^{\,i'}_{j'})\right)^2\geq d(d+1).
\end{equation}
Note that the lower bound in (\ref{WB2}) is achieved if and only if ${\vec{r}^{\,i}_j\cdot\vec{r}^{\,i'}_{j'}=0}$, for all $i\not=i'$, $j,j'\in\{1,d\}$. To conclude the proof, note that orthogonality in the Bloch hypersphere implies unbiasedness in the Hilbert space.
\end{proof}

Proposition \ref{prop2d_d+1} is interesting in the sense that, when restricted to the particular case of 2-designs formed by orthonormal bases, it resolves Zhu's conjecture. This is so because the remaining smaller case composed by $d$ bases is not possible. Indeed, A. Scott has shown that $d^2$ complex vectors in dimension $d$ form a 2-design if and only if they are a SIC-POVM~\cite{renes2004symmetric}. An intriguing open question is now how to generalize Proposition \ref{prop2d_d+1} to sets of $d(d+1)$ vectors beyond orthonormal bases. Recently, families of uniformly-weighted quantum state $2$-designs in dimension $d$ of size exactly
$d(d+1)$ that do not form complete sets of MUB were found, disproving Zhu's conjecture~\cite{iosue2024projective}.

 To conclude this section, let us mention an interesting property that connects cyclic $t$-designs with minimum uncertainty states.

\begin{prop}
Let $U$ be a unitary matrix producing a cyclic $t$-design composed of $k+1$ orthonormal bases. Then, each eigenvector of $U$ defines a minimum uncertainty state with respect to the following entropic uncertainty relation~\cite{sussman2007discrete}:
\begin{equation}\label{ARE}
\frac{1}{k+1}\sum_{j=1}^{k+1}H_j\geq\log_2(k+1)-1,
\end{equation}
where $H_j\!=\!-\log_2\qty[\sum_{\alpha=1}^d \qty(p^j_\alpha)^2]$, $p^j_\alpha=\abs{\ip*{\phi}{\psi^j_\alpha}}^2$, $\ket{\phi}$ is the state of the system, and $\ket{\psi^j_\alpha}$ is the $\alpha$-th column of $U^j$. 
\end{prop}

\begin{proof}
First, let us note that the left hand side of (\ref{ARE}) is minimized only if all Rényi entropies $H_j$ have the same value; see Section 4 in~\cite{sussman2007discrete}. Among all those cases, the lower bound established in (\ref{ARE}) is only achieved when $\qty{p^j_\alpha}$ comes from a $2$-design~\cite{sussman2007discrete}. To conclude the proof, we should prove that an eigenvector of $U$, called $\ket{\phi}$, produces identical Rényi entropies $H_j$. Indeed, in this case that probabilities $p^j_\alpha$ do not depend on the index $j$, i.e., 
\begin{equation*}
p^j_\alpha=\abs{\ip{\phi}{\psi^j_\alpha}}^2=\abs{\mel{\phi}{U^j}{\psi^1_\alpha}}^2=\abs{\ip{\phi}{\psi^1_\alpha}}^2,
\end{equation*}
for every $j=1,\dots,k+1$.
\end{proof}


\subsection{Approximate cyclic \texorpdfstring{$t$}{}-designs using random Hamiltonians}\label{res2}


Evolution of quantum states is generated by Hamiltonians $H$, and the choice of a suitable Hamiltonian is subject to experimental limitations and imperfections. One of the extreme cases is where one is able to control the eigenbasis, but not the exact energy levels. 
In order to consider how well powers of unitaries derived from Hamiltonians with random eigenvalues approximate cyclic designs, we first need to define what we mean by approximation. Therefore, we will use a computationally tractable definition of $\epsilon$-approximate $t$-design,

\begin{defn}\label{dfn:epsilon_des}
 A set of vectors $\qty{\ket{\psi_i}}_{i=1}^m$ defines an $\epsilon$-approximate $t$-design with error $\epsilon$ if 
 \begin{equation}
 \frac{1}{m^2}\sum_{i,j=1}^m \abs{\ip{\psi_i}{\psi_j}}^{2t} = \frac{1+\epsilon}{d_{sym}}, 
 \end{equation}
 where $d_{sym} = \binom{d+t-1}{t}$.
\end{defn}
This definition can be connected back to a definition in terms of $\infty$-norm introduced in~\cite{ambainis07} by Ambainis and Emerson,
\begin{thm}
 Consider an $\epsilon$-approximate $t$-design $\qty{\ket{\psi_i}}_{i=1}^m$. Then we find that
 \begin{equation}
 \norm{\frac{1}{m}\sum_{i=1}^m \op{\psi}^{\otimes t} \!\!-\!\! \int\dd{\ket{\psi}}\op{\psi}^{\otimes t}}_\infty \!\!\leq \delta,
 \end{equation}
 with $\delta = \sqrt{\epsilon}\frac{\sqrt{d_{sym}-1}}{d_{sym}}$.
\end{thm}

\begin{proof}
 Define an operator
 $$
 S\equiv \frac{1}{m}\sum_{i=1}^m \op{\psi}^{\otimes t},
 $$
 which by definition has support on the symmetric subspace with dimensionality $d_{sym}$. One can easily find that the frame potential is equal to
 \begin{equation}
 \frac{1}{m^2}\sum_{i,j=1}^m \abs{\ip{\psi_i}{\psi_j}}^{2t} = \Tr S^2 = \sum_{j=1}^{d_{sym}} \lambda_j^2,
 \end{equation}
 where $\{\lambda_j\}$ is the set of eigenvalues of the operator $S$.

 We consider the maximization problem of one of the eigenvalues, e.g. $\lambda_1$, under constraints
 \begin{align*}
 \Tr S = 1, &&
 \Tr S^2 = \frac{1 + \epsilon}{d_{sym}}.
 \end{align*}

 Carrying out the maximization (lower sign for minimization), one finds that
 \begin{align*}
 \lambda_1 = \frac{1 \pm \sqrt{\epsilon}\sqrt{d_{sym}-1}}{d_{sym}}
 \end{align*}
 \begin{align}
 \lambda_i = \frac{d_{sym} \mp\sqrt{\epsilon}\sqrt{d_{sym}-1}}{d_{sym}(d_{sym}-1)}.
 \end{align}

 It has been already shown that~\cite{scott2006tight}
 \begin{equation}
 \tilde{S} \equiv \int_{\mathcal{H}_d} \op{\psi}^{\otimes t} \dd{\psi} = \frac{1}{d_{sym}}\Pi_{sym},
 \end{equation}
 with $\Pi_{sym}$ the projection onto the completely symmetric subspace with all eigenvalues equal to $1$. It then follows that
 \begin{equation}
 \norm{S - \tilde{S}}_\infty = \frac{\sqrt{\epsilon}\sqrt{d_{sym}-1}}{d_{sym}} = \delta,
 \end{equation}
 which is the maximal value.
\end{proof}

Introduction of the above notion is enough to formulate the following theorem.

\begin{thm}[Approximate cyclic designs]
 Consider a unitary matrix $U$ of order $d$ defined by the eigendecomposition,
 \begin{equation}
 U = V\Lambda V^\dagger,
 \end{equation}
 such that the basis $V^\dagger$ yields, by decoherence,
 a simplex 2-design. Additionally, assume that the eigenvalues $\Lambda = \sum_\alpha \lambda_\alpha \op{\alpha}$ are taken as i.i.d. random variables from 
 flat measure over the unit circle in the complex plane or, equivalently, Haar measure of $\mathcal{U}(1)$.

 Therefore, the set $\qty{V \Lambda^i V^\dag = U^i}_{i=0}^k$, composed of powers of the unitary operation $U$ 
 provides
 a projective $\epsilon$-approximate $2$-design. The error $\epsilon$ is equal to 
 \begin{equation}
 \ev{\epsilon} = \frac{2(d-1)}{(k+1)},
 \end{equation}
 on average and
 becomes negligible for $k\gg d$.
\end{thm}
\begin{proof}
 The proof is given in Appendix~\ref{app:randH_anyD_proof}.
\end{proof}

In contrast, a random unitary matrix $U'$ taken from the Haar measure over $\mathcal{U}(d)$, would yield an $\epsilon$-approximate $2$-design, for which the approximation $\epsilon$ does not vanish with the number of bases $k$.


The above allows us to propose an operational interpretation of the error $\epsilon$ in terms of a tomographic scheme.

\begin{cor}
 Let us consider a cyclic projective $\epsilon$-approximate $2$-design $$\qty{U^j\ket{\mu} = e^{i j \tau_{int} H}\ket{\mu}}_{j=0,\mu=1}^{k,d}$$ of dimension $d$ and order $k$, defined by a Hamiltonian $H$ and characteristic interaction time $\tau_{int}$.
 Time $T$ necessary for
 full experimental implementation of a tomographic scheme of a state $\ket{\psi}$ with preparation time $\tau_{prep}$ with $N$ samples per each basis $U^j$
 is therefore equal to
 \begin{equation}
 T = N\qty[\frac{k(k+1)}{2}\tau_{int} + (k+1) \tau_{\text{prep}}].
 \end{equation}
 The approximate reconstruction formula yields 
 \begin{equation}
 \tilde{\rho}=
 \frac{1}{k+1} \sum_{j=0}^k \sum_{\mu = 1}^d \qty[p_{j,\mu} (d+1) - 1] \op{\psi_{j,\mu}}
 \end{equation}
 where $p_{j,\mu} = \ev{\rho}{\psi_{j,\mu}}$ being measurement probabilities on the target state $\rho$. Error from $\epsilon$-approximate $2$-design is bounded from above by
 \begin{equation}
 \norm{\rho - \tilde{\rho}}_\infty \leq d(d+1)\delta.
 \end{equation}
\end{cor}
Proof of the last equation is given in Appendix~\ref{app:random_reco}.

\subsection{Numerical results}\label{res3}

In this section, we introduce an algorithm used to obtain cyclic $t$-designs numerically. First, let us point out that the construction of complex projective $t$-designs is a challenging task. However, some constraints applied over the elements of the set allow us to simplify the construction process. In our approach, we restrict the $t$-designs to be sets of orthonormal bases. In addition, we simplify the construction of the bases by assuming that all of them are generated through powers of a single unitary transformation.

In our approach, we start by considering the generalized Gell-Mann matrices~\cite{georgi2000lie}, denoted as $\left\{ \lambda_j\right\}_{j=1}^{d^2-1}$, which can be used to parameterize any Hermitian matrix $H$ in dimension $d$. That is,
\begin{equation} \label{m_h}
H=\sum_{j=1}^{d^2-1}C_j\lambda_j.
\end{equation}
Thus, we generate a unitary matrix $U$ of order $d$ given by $U=e^{iH}$. Then, we look for a set of parameters $\{C_j\}_{j=1}^{d^2-1}$ such that the columns of the $k+1$ matrices $\{U^0,U^1,\dots,U^k\}$ conform a cyclic $t$-design. The $t$-design property is imposed by minimizing over the $d^2-1$ real parameters $C_1,\dots,C_{d^2-1}$. In other words, we minimize the frame potential, i.e., the LHS of Welch bound (\ref{WB}), as a function of the parameters $C_j$. 
\begin{equation*}
 \sum_{\ell,\ell'=0}^{k}\sum_{\beta,\beta'=1}^{d} |\langle \psi^{\ell}_{\beta}|\psi^{\ell'}_{\beta'} \rangle|^{2t} \geq \frac{[(k+1)d]^2}{\binom{d+t-1}{t}}, \quad t\in \mathbb{N} 
\end{equation*}
where upper index indicates the basis and lower index indicates the element within the basis.

Using this method, we found solutions for $k = 2,\dots, 18$ in dimension $d=2$. In the case of $d=3$, the sets of matrices form a complex protective $2$-design for $k=6,\dots, 13$, and in dimension $d=4$ there are solutions for $k=4, 6, 7, 10, 11, 14, 15$. The codes used for numerical calculations are available at~\cite{GitHubTest}.
In addition, results of the numerical search together with analytical results for existence of cyclic $2$-designs are summarized in Table~\ref{table2}.


\begin{table}
\begin{center}
\begin{tabular}{|c|c|c|c|}
\hline
\backslashbox{$k$}{$d$} & 2 & 3 & 4 \\ 
 \hline
2 &{\color{green!75!black}MUB}&-&-\\
 \hline
3 &{\color{gray} A}&{\color{red!75!black}\tikzxmark}&-\\
 \hline
4 &{\color{gray} A}&&{\color{green!75!black}MUB}\\
 \hline
5 &{\color{gray} A}&&\\
 \hline
6 &{\color{gray} A}&{\color{gray} A}&{\color{blue} N}\\
 \hline
7 &{\color{gray} A}&{\color{gray} A}&{\color{blue} N}\\
 \hline
8 &{\color{gray} A}&{\color{gray} A}&{\color{blue} N}\\
 \hline
9 &{\color{gray} A}&{\color{gray} A}&\\
 \hline
10 &{\color{gray} A}&{\color{gray} A}&{\color{blue} N}\\
 \hline
11 &{\color{gray} A}&{\color{gray} A}&{\color{blue} N}\\
 \hline
12 &{\color{gray} A}&{\color{gray} A}&{\color{gray} A}\\
 \hline
\end{tabular}
\end{center}
\caption{(Color online) Existence of cyclic $2$-designs composed of $k+1$ bases in dimension $d$ for small values of $k$ and $d$. Cyclic MUB are shown in green (MUB), non-existing cases in red (X), numerical solutions in blue (N), and analytical construction for $k\geq d(d-1)$ in gray (A). 
Not for all $k>d(d-1)$ there exist difference sets~\cite{du1998handbook} -- an example is $d=5, k=21$, where an exhaustive numerical search shows that a difference set does not exist.
White spaces denote unresolved cases.}
\label{table2}
\end{table}

\section{Discussion}
\label{sec:discussion}

In this study, we introduced the concept of \textit{cyclic measurements}, which are specialized sets of rank-one projective measurements interconnected by the powers of a single unitary transformation. Such a relation simplifies experimental implementation, making these measurements particularly appealing for practical use. More specifically, we examined a subset of cyclic measurements known as \emph{tight cyclic measurements} ($2$-designs) which are generated by powers of a single unitary operation applied to the states from the computational basis.


These measurements embody the mathematical framework of complex projective $t$-designs~\cite{scott2006tight}. From a practical standpoint, the tightness of these measurements significantly reduces propagation of statistical errors, thus offering a highly effective means for implementing quantum state tomography of density matrices in a simple and reliable manner.

We delineated necessary conditions for a unitary matrix to facilitate the creation of a cyclic tight measurement and demonstrated analytically that such $t$-designs are limited to $t <4$ for arbitrary dimension, together with solid evidence that stronger bound, $t<3$, holds in every dimension $d>2$; in contrast, for $d=2$ we find that $t=3$ is achievable. Our construction proves viable in any dimension that supports a specific type of difference sets, applicable across an infinite sequence of dimensions. Regarding limitations, our findings confirm that no single complex Hadamard matrix can be a part of a cyclic design in three-dimensional space.  In particular, this property explicitly rules out the possibility of cyclic mutually unbiased bases in this dimension. Although the existence of cyclic MUB for dimensions beyond powers of two is unclear, almost complete cyclic sets of $d$ MUB in prime dimension $d = p$ can be obtained using the standard construction by Wootters~\cite{wootters1989optimal}; in prime power dimensions, $d = p^n$, existence of $(d+1)/2$ mutually unbiased bases generated by powers of a single operator has been demonstrated in~\cite{Chau_2005_unconditionally}.

Conversely, our findings suggest that randomly sampling a system undergoing evolution with an appropriately selected Hamiltonian may produce measurements that approximate those necessary for implementing a complex projective 2-design effectively. Additionally, our numerical investigations indicate that in dimensions 3 and 4, cyclic designs consisting of the minimal possible number of elements do not emerge from the construction described in this work. This opens avenues for further research, particularly in exploring alternative constructions for low-dimensional cyclic $t$-designs.

\begin{acknowledgments}

DG acknowledges financial support from grant FONDECYT Regular no. 1230586, Chile. VGA belongs to the PhD Program \emph{Doctorado en F\'isica, menci\'on F\'isica-Matem\'atica}, Universidad de Antofagasta, Antofagasta, Chile, and acknowledges the support of the
ANID doctoral fellowship no. 21221008. 
JCz acknowledges the financial support from NCN
DEC-2019/35/O/ST2/01049 and a grant from
the Faculty of Physics, Astronomy and Applied Computer Science under the Strategic
Programme Excellence Initiative at Jagiellonian
University and is supported by the start-up grant of the Nanyang Assistant Professorship at the Nanyang Technological University in Singapore, awarded to Nelly Ng. JCz is grateful for the hospitality during his visit at Universidad de Antofagasta, supported by MINEDUC-UA project, code
ANT 1999. KŻ acknowledges support by
the DQUANT QuantERA II project no. 2021/03/Y/ST2/00193
that has received funding from the European Union’s Horizon 2020 research programme.
\end{acknowledgments}

\appendix

\section{Cyclic $2$-designs in dimension $2$ are $3$-designs}\label{app:proof_d2}

In dimension $d=2$,
it is natural to consider geometric states and transformation in the context of the Bloch ball. Therein, an orthonormal basis $\ket{\psi},\ket{\psi_\perp}$ is represented by a pair of antipodal points, with a single line that goes through both antipodes and center of the sphere. In addition, any unitary $U$ translates to a rotation around a given fixed axis.

Without loss of generality, we can choose the $z$ axis as the rotation axis, and a line forming an arbitrary angle $\theta$ with respect to it. That is, we can consider the line defined by the basis
\begin{align}\label{psis}
 \ket{\psi} & = \mqty(\cos\frac{\theta}{2} \\ \sin\frac{\theta}{2}), &
 \ket{\psi_{\perp}} & = \mqty(\sin\frac{\theta}{2} \\ -\cos\frac{\theta}{2}),
\end{align}
and a diagonal unitary matrix
\begin{equation}
 U= \begin{pmatrix}
 1 & 0 \\
 0 & e^{i\gamma}
 \end{pmatrix},
\end{equation}
with $\gamma = \frac{2\pi}{k+1}$. The first diagonal entry has been set to $1$ due to the freedom of a global phase, whereas the second diagonal entry is chosen such that $U^{k+1}=\mathbb{I}$. Thus, if we apply $U^{j}$, with $j = \qty{0,\hdots,k}$, to the states (\ref{psis}) we obtain a set composed by states of the form
\small\begin{equation}
{\begin{aligned}\label{basis}
 \ket{\psi_j} & = U^j\ket{\psi} = \mqty(\cos\frac{\theta}{2} \\ e^{i j \gamma}\sin\frac{\theta}{2}), \\
 \ket{\psi_{\perp j}} & = U^j\ket{\psi_\perp} = \mqty(\sin\frac{\theta}{2} \\ -e^{i j \gamma}\cos\frac{\theta}{2}).
\end{aligned}}
\end{equation} \normalsize
To be a cyclic $t$-design, the set has to saturate the Welch bound, which can be stated as
\begin{equation} \label{welch_b_d2}
 2\sum_{j,\ell=0}^{k} \qty[\left| \langle \psi_{j} | \psi_{\ell}\rangle \right|^{2t} + \left| \langle \psi_{j} | \psi_{\perp \ell}\rangle \right|^{2t}] = \frac{\left[ 2(k+1)\right]^2}{\binom{2+t-1}{t}}.
\end{equation}
From (\ref{basis}), we have
\begin{align}
 \left| \langle \psi_{j} | \psi_{\ell}\rangle \right|^{2}& = \left| \cos^2\frac{\theta}{2}+ e^{i\gamma(\ell-j)}\sin^2\frac{\theta}{2} \right|^2 \nonumber\\
 & = \left( 1- \sin^2\theta\,\sin^2\qty[\frac{\gamma}{2}(\ell-j)]\right),
\end{align}
and
\begin{align}
 \left| \langle \psi_{j} | \psi_{\perp \ell}\rangle \right|^{2} &= \left| \cos\frac{\theta}{2}\sin\frac{\theta}{2}+ e^{i\gamma(\ell-j)}\cos\frac{\theta}{2}\sin\frac{\theta}{2} \right|^2 \nonumber \\
 &= \left( \sin^2\theta\,\sin^2\frac{\gamma}{2}(\ell-j)\right).
\end{align}
Let $q_{j\ell}=- \sin^2\theta\,\sin^2\qty[\frac{\gamma}{2}(\ell-j)]$, equation (\ref{welch_b_d2}) reduces to
\begin{equation}
 2\sum_{j,\ell=0}^{k} (1+q_{j\ell})^t+(-q_{j\ell})^t=\frac{\left[ 2(k+1)\right]^2}{\binom{2+t-1}{t}}.
\end{equation}
Then, for $t=3$
\begin{equation} \label{wb_alpha}
 2\sum_{j,\ell=0}^{k} (1+q_{j\ell})^3+(-q_{j\ell})^3= 2\sum_{j,\ell=0}^{k} 1+3q_{j\ell}+3q_{j\ell}^2.
\end{equation}
And, given that
\begin{align}
 q_{j\ell} &=- \sin^2\theta\,\sin^2\qty[\frac{\gamma}{2}(\ell-j)] \nonumber\\ 
 & = -\frac{\sin^2\theta}{2}\left(1-\frac{e^{i\gamma(\ell-j)}+e^{-i\gamma(\ell-j)}}{2}\right),
\end{align}
we have
{\small\begin{equation}\label{sum_alpha}
\sum_{j,\ell=0}^{k}\!q_{j\ell}= -\!\sum_{j,\ell=0}^{k}\!\frac{\sin^2\theta}{2}\!\left(1-\frac{e^{i\gamma(\ell-j)}+e^{-i\gamma(\ell-j)}}{2}\right).
\end{equation}}

For the first term at the right hand side of (\ref{sum_alpha}), we have $(k+1)^2$ times the factor $\frac{\sin^2\theta}{2}$. On the other hand, the second term involves sums of the form
\begin{equation}
\sum_{m=0}^{k}Ce^{i\frac{2\pi}{k+1} m}=0,
\end{equation}
due the symmetry of the phases over the unit circle. 
In the same way, the third term in (\ref{wb_alpha}) reduces to
\begin{equation} \sum_{j,\ell=0}^{k}q_{j\ell}^2=-\frac{3}{8}(k+1)^2 \sin^4\theta.
\end{equation}
Inserting these expansions in the RHS of (\ref{wb_alpha}), we obtain
\begin{eqnarray*}
&&2\sum_{j,\ell=0}^{k} 1+3q_{j\ell}+3q_{j\ell}^2 = \\
&&2(k+1)^2 \left(1-\frac{3}{2}\sin^2\theta+\frac{9}{8}\sin^4\theta\right).
\end{eqnarray*}
In order to define a $3$-design in the simplex and, in turn, in the Hilbert space, it should be satisfied that 
$\cos^2\frac{\theta}{2}=\frac{3+\sqrt{3}}{6}$, implying that $\sin^2\theta=\frac{2}{3}$. Actually, this restriction provides the global minimum of the function
\begin{equation}
f(\theta)=1-\frac{3}{2}\sin^2\theta+\frac{9}{8}\sin^4\theta.
\end{equation}
To conclude, we emphasize that any other existing cyclic $t$-design is a rigid rotation of the ones generated above. Indeed, any other cyclic $t$-design will consider a rotation with respect to an axis different from $z$ in the Bloch sphere, defined by the eigenvectors of $U$.

\section{Full description of orthonormal bases in \texorpdfstring{$d = 3$}{} yielding \texorpdfstring{$2$}{}-designs in the simplex} \label{app:cyclic_d3}



Let $U$ be a unitary transformation in $d=3$ whose eigendescomposition is given by
 \begin{equation}
 U=V\Lambda V^\dagger,
 \end{equation}
 where $V^\dagger = \left( \ket{\psi_0}\, \ket{\psi_1}\, \ket{\psi_2}\right)$ is the matrix whose columns are the eigenvetors of $U$ and $\Lambda=\text{diag}(\lambda_{0},\lambda_{1},\lambda_{2})$.

 Now, let us focus on the number of parameters that are actually free in the choice of the basis $V^\dagger$ of the operator, which at the outset seems to have $9$ free parameters. However, we may decompose $V^\dagger = D_1 \tilde{V}$, with $D_1$ being a diagonal unitary matrix and $\tilde{V}$ a unitary matrix with dephased columns. In this form, the entries $\tilde{V}_{\alpha1}$ are real, leaving the matrix with $6$ real parameters. The enphasing matrix $D_1$ is irrelevant as $D_1^\dagger \Lambda D_1 = \Lambda$. Thus, we may take the basis $V^\dagger = \tilde{V}$, without loss of generality.
 
 Next, we will focus on the probability vectors $\vb{p}^{(\beta)}$ with entries given by $p^{(\beta)}_\alpha = \abs{\ip{\alpha}{\psi_\beta}}^2$. It will be useful to consider the projection of the above probability vectors onto the 2-dimensional regular simplex, which may be achieved using a projection operator
 \begin{equation}
 \mathbb{P} = \mqty(
 \frac{1}{\sqrt{2}} & \frac{-1}{\sqrt{2}} & 0 \\
 \frac{1}{\sqrt{6}} & \frac{1}{\sqrt{6}} & \frac{-2}{\sqrt{6}}).
\end{equation}
This allows us to cast any vector $\vb{p}^{(\beta)}$ to a $2$-dimensional vector
\begin{equation}
 \vb{W}_\beta = \mathbb{P}\vb{p}^{(\beta)}=\left(\frac{p_{0}^{(\beta)}-p_{1}^{(\beta)}}{\sqrt{2}},\frac{p_{0}^{(\beta)}+p_{1}^{(\beta)}-2p_{2}^{(\beta)}}{\sqrt{6}}\right).
\end{equation}

It is easy to calculate that vertices of the probability simplex $\Delta_2$ in this representation are given by points
$\boldsymbol{\delta}^1 = \left(1/\sqrt{2},1/\sqrt{6}\right), \boldsymbol{\delta}^2 =\left(-1/\sqrt{2},1/\sqrt{6}\right), \boldsymbol{\delta}^3 =\left(0,-2/\sqrt{6}\right)$.\\

Since obtaining a general solution seems to be a formidable task, we will proceed following the methods proposed in~\cite{Hammer1956NumericalIO, kuperberg2004numerical}. Let us consider rescaled versions of the points $\vb{W}^i = r \boldsymbol{\delta}^i$. Additionally, by imposing the requirements on the averages of monomials of degree 1 and 2 to be equal to these over the full simplex,
\begin{equation*}
 \langle p_\alpha \rangle = \frac{1}{3},\,
 \langle p_\alpha^2 \rangle = \frac{1}{6},\,
 \langle p_\alpha p_\beta \rangle = \frac{1}{12},\\
\end{equation*}
we find that the scaling factor has to be set to $r = \frac{1}{2}$. Furthermore, it is easily verified that the requirements on the 2-design remain fulfilled under a rotation
 \begin{equation*}
 R(\phi) = \mqty(
 \sin(\phi) & -\cos(\phi) \\
 \cos(\phi) & \sin(\phi) \\
 ),
 \end{equation*} 
such that if $\{\vb{W}^0,\,\vb{W}^1,\,\vb{W}^2\}$ provides a simplex 2-design, then $\{R(\phi)\vb{W}^0,\,R(\phi)\vb{W}^1,\,R(\phi)\vb{W}^2\}$ remains a simplex 2-design. This introduces a free parameter $\phi$ within the amplitudes of the state vectors. With this, we may go back to the full amplitudes, which are given by
\begin{subequations}
\label{eq:amplitudes}
\begin{align}
 a_0^2 & = \frac{1}{3}\qty(1 - \cos\phi), \\
 a_1^2 & = \frac{\sin (\phi )}{2 \sqrt{3}}+\frac{\cos (\phi )}{6}+\frac{1}{3}, \\
 a_2^2 & = - \frac{\sin (\phi )}{2 \sqrt{3}}+\frac{\cos (\phi )}{6}+\frac{1}{3}.
\end{align}
\end{subequations}

With the above, the entire structure of the eigenbasis is given by
 \begin{equation}
 V^\dagger = \mqty(
 a_{0} &a_{2} & a_{1}\\
 a_{1}e^{i\gamma_{0}} &a_{0} & a_{2}e^{i\gamma_{1}}\\
 a_{2}e^{i\gamma_{2}} & a_{1}e^{i\gamma_{3}}& a_{0}
 ),
 \end{equation}
 
 Solving for orthogonality, $VV^\dagger = \mathbb{I}$, we find that the phases $\gamma_0$ through $\gamma_3$ are given by
 \begin{widetext}
 \begin{subequations}
 \label{eq:phases}
 \begin{align}
 \gamma_{0} &=\tan^{-1}\left(\frac{8 \cos (\phi ) (\cos (\phi )+1)-7}{\sqrt{-48 \cos (3 \phi )-33}}\right)= \gamma_{1},\\
 \gamma_2&= \tan^{-1}\left(\frac{\sin \left(\frac{\phi }{2}\right) \left(\sqrt{3} \cot \left(\frac{\phi }{2}\right)-1\right) \left(4 \cos (\phi )+2 \cos (2 \phi )+4 \sqrt{3} \sin (\phi ) (\cos (\phi
 )-1)+3\right)}{2 \cos\left(\frac{1}{6} (3 \phi +\pi )\right) \sqrt{-48 \cos (3 \phi )-33}}\right),\\
 \gamma_3 &= \tan ^{-1}\left(\frac{3+4\cos(\phi)+2\cos(2\phi)+4\sqrt{3}(\cos(\phi)-1)\sin(\phi)}{\sqrt{-33-48\cos(3\phi)}}\right).
 \end{align}
 \end{subequations}
 \end{widetext}

  Additionally, since all the phases are real, we have $\gamma_i\in\mathbb{R}$, we find additional restriction, $\phi\in\left(\phi_- + n\frac{2\pi}{3}, \phi_++ n\frac{2\pi}{3}\right)$ where the parameter $\phi_\pm = \pm \cos ^{-1}\left(\frac{1}{8} \left(-1-3 \sqrt{5}\right)\right)$ and $n\in\mathbb{N}$.
 
 Particular examples of the bases from this family of solutions include the $\phi = \pi$ case, aligned with the simplex,
 
 \begin{equation}
 V^\dagger = \frac{1}{\sqrt{6}}\mqty(
 2 & 
 q_1 & 
 q_1 \\
 1 & 
 2 & 
 q_2\\
 1 & 
 q_2& 
 2 \\
 ),
 \end{equation}
 with 
 \begin{equation*}
     q_{1}=\qty(-\frac{7}{8} \pm i \sqrt{\frac{15}{64}}),\quad
     q_{2}=\qty(-\frac{1}{4} \mp i \sqrt{\frac{15}{16}}).
 \end{equation*}
 
 On the other hand, if we set the parameter as $\phi = \arccos \left(\frac{1}{8} \left(-1-3 \sqrt{5}\right)\right)$ we obtain 
 \begin{equation}
 V^\dagger= \frac{1}{2}\mqty(
 \varphi & \varphi^{-1} & 1 \\
 -1 & \varphi & \varphi^{-1} \\
 -\varphi^{-1} & -1 & \varphi
 ),
 \end{equation}
 where $\varphi = \frac{1 + \sqrt{5}}{2}$ is the golden ratio. Thus, it is legitimate to call $V^\dagger$ the \emph{golden orthogonal matrix} of order three and the corresponding basis the real \emph{golden basis}.
 



\subsection{Note on 3-point designs in $d=3$}\label{app:baladram_invalidation}

Taking into account the general form of amplitudes \eqref{eq:amplitudes}, which yield all possible 3-point configurations in $\Delta_3$ which form simplex 2-designs, we may now consider a problem of 3-design in $\Delta_3$. To do this, we consider a cost function of the form
    \begin{align}
        F(\phi) = & \sum_\beta \qty(\ev{p_\alpha^3}-\frac{1}{10})^2 + \sum_{\alpha\neq \beta} \qty(\ev{p_\alpha^2p_\beta}-\frac{1}{30})^2 \nonumber\\& + \sum_{\alpha\neq \beta\neq \mu\neq \alpha} \qty(\ev{p_\alpha p_\beta p_\mu}-\frac{1}{60})^2 \nonumber \\
        = & \frac{\sin^{2}{\left(\phi \right)}}{46656} + \frac{\sqrt{6} \sin{\left(\phi \right)}}{19440} + \frac{59641}{1555200}
    \end{align}
with averages taken with respect to the 3-point 2-designs, given implicitly as functions of the angle $\phi$. It is immediate to see that this equation has no solutions for $\phi\in\mathbb{R}$, and thus shows that there exist no 3-point 3-designs in $\Delta_3$. This stands in contrast with results presented in \cite{baladram2018explicit}. In fact, it invalidates Theorem 3.2 presented therein, concerning construction of simplex $t$-designs by utilizing cyclic permutations exclusively by providing an explicit counterexample.

Let us note that the statements in \cite{baladram2018explicit} can be corrected; however, as such correction goes beyond the scope of present work, it is deferred to \cite{czartowski2025comment}.


\vspace{0.2cm}
\section{Non-existence of cyclic MUB in \texorpdfstring{\mbox{$d = 3$}}{}}\label{App:noMUBd3}


Let us start emphasizing that the angle $\phi$ from Eqs.\eqref{eq:amplitudes}-\eqref{eq:phases} belongs to the range \mbox{$\phi\in\qty[\arccos \left(\frac{1}{8} \left(-1-3 \sqrt{5}\right)\right),\pi]$}. Taking this into account, let us demonstrate the non-existence of cyclic MUB in dimension $d=3$ by considering a stronger claim: there is no complex Hadamard matrix of order 3 of the form
\begin{equation}\label{H3}
 H = V(\phi)\Lambda V(\phi)^\dagger,
\end{equation}
where, without loss of generality, we may choose the eigenvalues of $H$ as 
$$\Lambda = \mqty(\dmat[0]{1,e^{2iw},e^{i(w-z)}}).$$ 
In order to prove the non-existence {of $H$} we consider a figure of merit {$F$} based on the diagonal entries of the matrix, that is,
\begin{equation}\label{FH}
 F\qty[H] = \sum_\alpha \qty(\abs{H_{\alpha\alpha}}^2 - \frac{1}{3})^2.
\end{equation}\vspace{0.1cm}

Note that $F[H] = 0$ is a necessary condition to produce a Hadamard matrix $H$. Thus, the goal of the proof is to show that it cannot be achieved using the eigenbasis that yields a simplex 2-design for any eigenvalues.
From combining (\ref{H3}) and (\ref{FH}) it is simple to show that
\begin{widetext}
\begin{eqnarray*}
F[H]&=&
\frac{1}{216} (-8 \cos (3 \phi ) (\cos (2 w) (4 \cos (w) \cos (z)-2 \cos (2 z)+1)-\cos (4 w)+\cos (2 z)-3)\\
&&+44 \cos (w-z)+19 \cos (2 (w-z))+8 \cos (3 w-z)+44 \cos (w+z)+19 \cos (2 (w+z)) \\
&&+8 \cos (3 w+z)+44 \cos (2 w)+19 \cos (4 w)+8 \cos (2 z)+75).
\end{eqnarray*}
\end{widetext}

This expression can be seen as a polynomial of order at most $4$ in trigonometric functions, and the same occurs for its derivatives. It is simple to check that the extreme conditions
\begin{equation}
 \begin{aligned}
 \pdv{F}{w} & = 0,&
 \pdv{F}{z} & = 0,&
 \pdv{F}{\phi} & = 0,
 \end{aligned}
\end{equation}

have $2\times4=8$ solutions for cosines, where we skip the multiplicity corresponding to $\phi$, since we consider it within a range that cancels the multiplicities. This number is duplicated to $16$ due to the parity of cosines. After considering in addition the standard condition on minima and factoring out the symmetries, we end up with two $(w,z,\phi)$ points:
{\footnotesize\begin{align*}
 \qty(\arctan(\sqrt{\frac{6}{5}}),\, \pi -\arctan\left(\frac{9 \sqrt{\frac{6}{5}}}{13}\right),\, \pi) \rightarrow F = \frac{68}{3993} \\
 \qty(\arctan(2\sqrt{30}),\,\pi,\,\pi) \rightarrow F = \frac{68}{3993}.
\end{align*}}

For completeness, we have to consider the minima with respect to $w,z$ variables at the boundary $\phi = \phi_{\text{min}}$. In this case we find additional two points
{\small \begin{align*}
 \qty(\frac{2\pi}{3},0,\phi_{min}) \rightarrow F = \frac{1}{48 } \\ \qty(\arctan(2\sqrt{2}),\arctan(2\sqrt{2}),\phi_{min}) \rightarrow F = \frac{4}{27}.
\end{align*}}
Given that all such extreme solutions for $F$ are greater than 0, we conclude that matrix $H$ does not exist in dimension $3$. This strong result implies in turn non-existence of full set of MUB as an element of any cyclic $t$-design. Therefore, $F[H]>0$ for all $\phi,w,z$, which implies non-existence of a Hadamard matrix fulfilling the necessary conditions for generating a cyclic design.\qed

\begin{widetext}
\section{Proof of Theorem \ref{thm:deg_less_4} and conjecture for lower bounds on the size of simplex designs}\label{app:rank_M_matrix}
Let us consider matrix $M^{\alpha\beta}_{\mu\nu} = \ev{p_\alpha p_\beta p_\mu p_\nu}_{\Delta_d}$ of moments over a $d$-point simplex and consider set of equations
\begin{equation}
    \sum_{\mu\nu=1}^d M^{\alpha\beta}_{\mu\nu}x_{\mu\nu} = 0
\end{equation}
for all $\mu,\nu$. Due to the symmetry of $M$ it is immediate to find that they are satisfied every for all antisymmetric cases, ie. $x_{\mu\nu} = - x_{\nu\mu}$ -- with this we can restrict our considerations to $d(d+1)/2$-dimensional symmetric space. Thus, we consider equations of the form
\begin{equation}
    \sum_{\mu\nu=1}^d M^{\alpha\beta}_{\mu\nu} y_{(\mu\nu)} = 0
\end{equation}
where $(\mu\nu)$ are ordered pairs of indices. Additionally, we note that if a given set $y_{(\mu\nu)} = a_{(\mu\nu)}$ provides a solution, then by symmetry of the problem $y_{(\pi(\mu)\pi(\nu))} = a_{(\mu\nu)}$ is also a solution for any relabelling $\pi\in\mathcal{S}_d$. Thus, it is enough to consider cases when $y_{(\mu\nu)} = y_{(\pi(\mu)\pi(\nu))} = a_{(\mu\nu)}$. If we define $A_1 = y_{(\mu\mu)}$ and $A_2 = y_{(\mu\nu)}$ for $\mu\neq \nu$, the problem is reduced to
\begin{equation}
    \sum_{\mu=1}^d M^{\alpha\beta}_{\mu\mu} A_1 + \sum_{\mu\neq \nu}^d M^{\alpha\beta}_{\mu\nu} A_2 = 0.
\end{equation}
From this we find at most two equations, stemming from $\alpha=\beta$ and $\alpha\neq \beta$ cases respectively
\begin{alignat}{2}
    \qty(\ev{p_\alpha^4}+ (d-1) \ev{p_\alpha^2p_\beta^2})A_1 +& (d-1)\qty(2\ev{p_\alpha^3 p_\beta} + (d-2) \ev{p_\alpha^2p_\beta p_\mu})A_2 &= &0 \\
    \qty(2\ev{p_\alpha^3p_\beta}+ (d-2) \ev{p_\alpha^2p_\beta p_\mu})A_1 +& \left(2\ev{p_\alpha^2 p_\beta^2} + 4(d-2)\ev{p_\alpha^2p_\beta p_\mu} \right. &&\\
    +&  \left.(d-2)(d-3) \ev{p_\alpha p_\beta p_\mu p_\nu}\right)A_2 &=& 0
\end{alignat}
By using \eqref{eq:general_simplex_av} and performing elementary simplification we find the following equations explicitly
\begin{align}
    A_1 + \frac{(d-1)(d+2)}{3(d+3)} A_2 & = 0, \\
    A_1 + \frac{d^2+3d-2}{2(d+4)}A_2 & = 0.
\end{align}
The above two equations have only trivial solution, ie. $A_1 = A_2 = 0$. Hence, $\operatorname{rank}(M) = \binom{d+1}{2}$. \qed

Additionally, we believe that the following statement, which goes beyond the main focus of this manuscript, may be true in general
\begin{conj}
    Minimal number of elements $N_*(d,t)$ in a $t$-design in a $d$-point simplex $d$ is lower-bounded as
    \begin{equation}
        N_*(d,t) \geq \binom{d+\frac{t}{2}-1}{\frac{t}{2}}.
    \end{equation}
\end{conj}

Note that similar argument as for $t = 4$ can be put forward for arbitrary even $t$ by considering the rank of a matrix $M$ given as
\begin{equation}
    M^{\alpha_1\hdots \alpha_t}_{\beta_1\hdots \beta_t} = \ev{\prod_{k=1}^t p_{\alpha_k}p_{\beta_k}}.
\end{equation}
One could then use fully analogous symmetry as above to reduce the dimensionality of the problem from $d^t$ variables down $T$ variables, where $T$ is the number of ways in which the integer $t/2$ can be partitioned. Even though the proof strategy put above is fully clear, we do not see a way to prove that the resulting system of equations is infeasible for every even $t$, and it does not extend to odd degrees $t$.

\section{Cyclic designs based on \texorpdfstring{$\lambda = 1$}{} difference sets} \label{app:cd_anyD_proof}

Let us consider a set $\qty{U^j}_{j=0}^k$, where $U$ is a unitary matrix with eigendecomposition $U = V \Lambda V^\dagger$, and $\Lambda$ is the diagonal matrix of eigenvalues.
\begin{equation}
 \Lambda = \sum_\alpha \lambda_\alpha\op{\alpha} = \sum_\alpha e^{i\mu_\alpha}\op{\alpha},
\end{equation}
where we use an Ansatz $\mu_\alpha = \frac{2\pi}{k+1} N_\alpha$ with $N_\alpha \in \mathbb{N}$.
Furthermore, let us write the matrix $V^\dagger$ as set of column vectors,
\begin{equation}
 V^\dagger = \mqty(\ket{\psi_1} & \hdots & \ket{\psi_d}).
\end{equation}
Using $V^\dagger$, we can rotate the entire $2$-design such that $\qty{ U^iV^\dagger =  V^\dagger\Lambda^i}$. In this way, the problem is reduced to considering the vectors $\ket{\psi_\beta}$ and powers of the eigenvalues $\lambda_\beta$. 

Furthermore, let us assume that the vectors $\qty{\ket{\psi_\beta}}$ provide a simplex 2-design after projection onto the computational basis $\qty{\ket{\alpha}}$, which is a necessary but not sufficient condition to even consider the matrix $U$ to be capable of producing a cyclic 2-design, that is 

\begin{alignat}{5}
 \frac{1}{d} \sum_{\beta = 1}^d \abs{\braket{\psi_\beta}{\alpha}}^2 & = &
 \ev{p_\alpha}_\Delta & = \frac{1}{d} && = \binom{d}{1}^{-1}, \label{eq:simplex_design_1}\\
 \frac{1}{d} \sum_{\beta = 1}^d \abs{\braket{\psi_\beta}{\alpha}}^4 & = &
 \ev{p^2_\alpha}_\Delta & = \frac{2}{d(d+1)} && = \binom{d+1}{2}^{-1}, \label{eq:simplex_design_2}\\ 
 \frac{1}{d} \sum_{\beta = 1}^d \abs{\braket{\psi_\beta}{\alpha}\braket{\nu}{\psi_\beta}}^2 & = &
 \ev{p_\alpha p_\nu}_\Delta & = \frac{1}{d(d+1)} && = \frac{1}{2}\binom{d+1}{2}^{-1}. \label{eq:simplex_design_3}
\end{alignat}

Finally, let us assume a nontrivial relation that all the eigenvalues are different from each other, $N_\alpha \neq N_{\alpha'}$.
We start by considering the Welch bound for $t=1$,

\begin{align}
 \sum_{\beta,\beta'=1}^{d}\sum_{n,n'=0}^{k} \abs{\mel{\psi_{\beta'}}{\Lambda^{n-n'}}{\psi_\beta}}^2 
 & = \sum_{\beta,\beta'=1}^{d}\sum_{n,n'=0}^{k} \abs{\sum_{a = 1}^de^{i(n-n')\mu_{\alpha}}\braket{\psi_{\beta'}}{\alpha}\braket{\alpha}{\psi_\beta}}^2 \\
 & = \sum_{\beta,\beta'=1}^{d}\sum_{n,n'=0}^{k} \sum_{\alpha,\alpha' = 1}^d e^{i(n-n')\qty(\mu_{\alpha} - \mu_{\alpha'})}\braket{\psi_{\beta'}}{\alpha}\braket{\alpha}{\psi_\beta}\braket{\psi_{\beta}}{\alpha  '}\braket{\alpha'}{\psi_{\beta'}} \\
 & = \sum_{\substack{\beta,\beta'=1,\\ \alpha,\alpha' = 1}}^{d}\underbrace{\qty(\sum_{n,n'=0}^{k} e^{i(n-n')\qty(\mu_{\alpha} - \mu_{\alpha'})})}_{(*)}\braket{\psi_{\beta'}}{\alpha}\braket{\alpha}{\psi_\beta}\braket{\psi_{\beta}}{\alpha'}\braket{\alpha'}{\psi_{\beta'}}.
\end{align}
We focus on the underbraced term $(*)$ and leverage the Ansatz,

\begin{equation}
 (*) = \sum_{n,n'=0}^{k} e^{i\frac{2\pi}{k+1}(n-n')(N_\alpha - N_{\alpha'})} = (k+1) \sum_{m=0}^{k} e^{i\frac{2\pi}{k+1}m(N_\alpha - N_{\alpha'})} = (k+1)^2 \delta_{\alpha\,\alpha'}.
\end{equation}
Using this, we obtain
\begin{equation}
 \sum_{\beta,\beta'=1}^{d}\sum_{n,n'=0}^{k} \abs{\mel{\psi_{\beta'}}{\Lambda^{n-n'}}{\psi_\beta}}^2 = (k+1)^2\sum_{\beta,\beta'=1}^{d}\sum_{\alpha = 1}^d\abs{\braket{\psi_{\beta'}}{\alpha}\braket{\alpha}{\psi_\beta}}^2 = (k+1)^2 d,
\end{equation}
where in the last step we have used the fact that vectors $V^\dagger$ form a $2$-design in the probability simplex. Notice that the value $(k+1)^2d$ saturates the Welch bound~\eqref{welch_b_d2}, thus proving that the configuration is a complex projective $1$-design. Now, we shift to calculation of Welch bound for $t=2$. That is,

\begin{align*}
 & \sum_{\beta,\beta'=1}^{d}\sum_{n,n'=0}^{k} \abs{\mel{\psi_{\beta'}}{\Lambda^{n-n'}}{\psi_\beta}}^4 
 = \sum_{\beta,\beta'=1}^{d}\sum_{n,n'=0}^{k} \abs{\sum_{\alpha = 1}^d e^{i(n-n')\mu_{\alpha}}\braket{\psi_{\beta'}}{\alpha}\braket{\alpha}{\psi_\beta}}^4 = \\
  & \sum_{\substack{\beta,\beta'=1,\\ \alpha,\alpha',\nu,\nu' = 1}}^{d}\underbrace{\qty(\sum_{n,n'=0}^{k} e^{i(n-n')(\mu_{\alpha} + \mu_{\nu} - \mu_{\alpha'} - \mu_{\nu'})})}_{(**)}
 \braket{\psi_{\beta'}}{\alpha}\braket{\alpha}{\psi_\beta}\braket{\psi_{\beta}}{\alpha'}\braket{\alpha'}{\psi_{\beta'}} \braket{\psi_{\beta'}}{\nu}\braket{\nu}{\psi_\beta}\braket{\psi_{\beta}}{\nu'}\braket{\nu'}{\psi_{\beta'}},
\end{align*}
where we again look closer at the underbraced expression $(**)$,

\begin{align}
 \sum_{n,n'=0}^{k} e^{i\frac{2\pi}{k+1}(n-n')(N_\alpha + N_\nu - N_{\alpha'} - N_{\nu'})} & = (k+1) \sum_{m=0}^{k} e^{i\frac{2\pi}{k+1}m(N_\alpha + N_\nu - N_{\alpha'} - N_{\nu'})} \\
 & = \begin{cases}
 (k+1)^2 & N_\alpha + N_\nu - N_{\alpha'} - N_{\nu'} = 0\,\text{mod}\,(k+1) \\
 0 & \text{otherwise}.
 \end{cases}
\end{align}

There are three trivial cases, satisfied independently from the values $N_\alpha$: $\alpha = \nu = \alpha' = \nu'$, $\alpha = \alpha'\neq \nu = \nu'$ and $\alpha = \nu' \neq \alpha' = \nu$. We start by considering the first case
\begin{align}
 & \sum_{\beta,\beta'=1}^{d}\sum_{\alpha = 1}^d \abs{\braket{\psi_{\beta'}}{\alpha}\braket{\alpha}{\psi_\beta}}^4 = d^2\sum_{\alpha = 1}^d \qty(\frac{1}{d}\sum_{\beta=1}^{d}\abs{\braket{\psi_{\beta}}{\alpha}}^4)^2 = \frac{4d}{(d+1)^2},
\end{align}
where $(k+1)^2$ factor has been skipped for simplicity and we have used equation~\eqref{eq:simplex_design_2}. 
We consider the two remaining cases together, as they reduce to the same equation,
\begin{align}
 \sum_{\beta,\beta'=1}^{d}\,\sum_{\substack{\alpha,\nu=1\\ (\alpha\not= \nu)}}^d \abs{\braket{\psi_{\beta'}}{\alpha}\braket{\alpha}{\psi_\beta}}^2\abs{\braket{\psi_{\beta'}}{\nu}\braket{\nu}{\psi_\beta}}^2 = 
 d^2\sum_{\substack{\alpha,\nu=1\\ (\alpha\not= \nu)}}^d \qty(\frac{1}{d}\sum_{\beta=1}^d \abs{\braket{\psi_{\beta}}{\alpha}\braket{\nu}{\psi_\beta}}^2)^2 = \frac{d(d - 1)}{(d+1)^2},
\end{align}
where we again skipped $(k+1)^2$ term for simplicity. 
Collecting the three cases together, we have

\begin{equation}
 (k+1)^2d\frac{2d - 2 + 4}{(d+1)^2} = \frac{2(k+1)^2d}{d+1} = \frac{\qty[d (k+1)]^2}{\binom{d+1}{2}},
\end{equation}
which exactly saturates the Welch bound~\eqref{welch_b_d2} for $t=2$. Thus, the only condition for the set of the eigenvalues $\qty{\mu_\eta = \frac{2\pi}{k+1}N_\eta}$ to generate a cyclic 2-design from a basis $V^\dagger$ providing 2-design in the probability simplex is that there the differences $N_\alpha + N_\nu - N_{\alpha'} - N_{\nu'} = 0\;\text{mod} (k+1)$ are zero only for trivial sets of indices. This leads to difference sets with signature $\qty(k+1,\,d,\,1)$ being of interest. Apart from the greedy construction of the Mian-Chowla sequence, for any $d$ one can construct directly the set $\qty{2^i}_{i=0}^{d-1}$, which provides a difference set $(2^d, d, 1)$ for every $d$; it is easily proved as $2^i - 2^j = 2^{i'} - 2^{j'} \text{mod}\,2^d$ only for $i=i'$ and $j = j'$.

\section{Cyclic designs based on random Hamiltonians} \label{app:randH_anyD_proof}
Let us revisit a set $\qty{U^j}_{j=0}^k$ with $U = V \Lambda V^\dag$ 
where $\Lambda = \text{diag}\qty(\lambda_1,\hdots,\lambda_d) = \text{diag}\qty(e^{i\mu_1},\hdots,e^{i\mu_d})$
    . However, this time, let us take the numbers $\mu_\alpha$ from a uniform distribution over the interval $[0,2\pi)$. We will continue working in the basis using $V$ as rotation, $\qty{ U^i V^\dagger= V^\dagger\Lambda^i }$, so that the problem is reduced to considering the eigenvectors $\ket{\psi_\beta}$ and powers of the eigenvalues $\lambda_\beta$. Moreover, we will keep using the assumption that $V^\dagger$ provides by decoherence a 2-design in the simplex, fulfilling equations~(\ref{eq:simplex_design_1} - \ref{eq:simplex_design_3}).
    
Let us focus now on calculating the Welch bound~\eqref{welch_b_d2} for $t=2$. 

\begin{align}
 & \sum_{\beta,\beta'=1}^{d}\sum_{n,n'=0}^{k} \abs{\mel{\psi_{\beta'}}{\Lambda^{n-n'}}{\psi_\beta}}^4 
 = \sum_{\beta,\beta'=1}^{d}\sum_{n,n'=0}^{k} \abs{\sum_{\alpha = 1}^de^{i(n-n')\mu_{\alpha}}\braket{\psi_{\beta'}}{\alpha}\braket{\alpha}{\psi_\beta}}^4 = \\
 = & \sum_{\beta,\beta'=1}^{d}\sum_{\substack{\alpha,\alpha'=1,\\\nu,\nu' = 1}}^d \sum_{n,n'=0}^{k} e^{i(n-n')(\mu_{\alpha} + \mu_{\nu} - \mu_{\alpha'} - \mu_{\alpha'})}
 T_{\alpha\alpha'\nu\nu'}^{\beta\beta'} \\ 
 = & \sum_{\beta,\beta'=1}^{d}\sum_{\substack{\alpha,\alpha'=1,\\\nu,\nu' = 1}}^d \qty(\underbrace{(k+1)}_{\sum_{n=n'}\hdots} + \sum_{n\neq n'}^{k} e^{i(n-n')(\mu_{a} + \mu_{b} - \mu_{a'} - \mu_{b'})})
 T_{\alpha\alpha'\nu\nu'}^{\beta\beta'},
\end{align}
where for convenience we defined
$$
 T_{\alpha\alpha'\nu\nu'}^{\beta\beta'} = \braket{\psi_{\beta'}}{\alpha}\braket{\alpha}{\psi_\beta}\braket{\psi_{\beta}}{\alpha'}\braket{\alpha'}{\psi_{\beta'}} \braket{\psi_{\beta'}}{\nu}\braket{\nu}{\psi_\beta}\braket{\psi_{\beta}}{\nu'}\braket{\nu}{\psi_{\beta'}}.
$$

Separating the $n=n'$ case, we see that each sum over $\alpha,\alpha',\nu,\nu'$ acts only on the projectors on the basis vectors, leading to expressions

\begin{equation}
 \sum_\alpha \braket{\psi_\beta}{\alpha}\braket{\alpha}{\psi_{\beta'}} = \delta_{\beta\beta'},
\end{equation}
and thus
\begin{equation}
 (k+1)\sum_{\beta,\beta'=1}^{d}\sum_{\substack{\alpha,\alpha'=1,\\ \beta,\beta' = 1}}^d 
 T_{\alpha\alpha'\nu\nu'}^{\beta\beta'} = (k+1)d.
\end{equation}
On the other hand, expression of the form
$
 e^{i(n-n')(\mu_{\alpha} + \mu_{\nu} - \mu_{\alpha'} - \mu_{\nu'})}
$
has been already considered in Appendix~\ref{app:cd_anyD_proof}, where the values for trivial zeros of the exponent, $\alpha=\nu=\alpha'=\nu'$, $\alpha=\alpha',\nu=\nu'$ and $\alpha=\nu',\nu=\alpha'$, have been calculated. The remaining terms will be killed by averaging with respect to the phases $\qty{\mu_\eta}$. Thus, we find that 
\begin{equation}
 \ev{\sum_{\beta,\beta'=1}^{d}\sum_{\substack{\alpha,\alpha'=1,\\ \nu,\nu' = 1}}^d \sum_{n\neq n'}^{k} e^{i(n-n')(\mu_{\alpha} + \mu_{\nu} - \mu_{\alpha'} - \mu_{\nu'})}
 T_{\alpha\alpha'\nu\nu'}^{\beta\beta'}} = k(k+1) \frac{d^2}{\binom{d+1}{2}},
\end{equation}
where for brevity we define the average over the random eigenvalues,
$$
\ev{F} = \qty(\frac{1}{2\pi})^d\int_{0}^{2\pi}\dd{\mu_1}\hdots\dd{\mu_d} F.
$$

Putting the two expressions together allows us to conclude the calculation,

\begin{equation}
 \ev{\sum_{\beta,\beta'=1}^{d}\sum_{\substack{\alpha,\alpha'=1,\\ \nu,\nu' = 1}}^d \sum_{n, n' = 1}^{k} e^{i(n-n')(\mu_{\alpha} + \mu_{\nu} - \mu_{\alpha'} - \mu_{\nu'})}
 T_{\alpha\alpha'\nu\nu'}^{\beta\beta'}} = \frac{\qty[(k+1)d]^2}{\binom{d+1}{2}} + \frac{(d-1) d (k+1)}{d+1}.
\end{equation}
This already shows that for large orders $k$ unitary operations with properly chosen eigenbasis and random eigenvalues taken independently from the uniform distribution over $[0,2\pi)$ interval, acting on the computational basis, converge to cyclic 2-designs. It is more readily visible after transforming the result to the form congruent with~\eqref{WB}
\begin{align}
 \frac{1}{\qty[d(k+1)]^2}\ev{\sum_{\beta,\beta'=1}^{d}\sum_{\substack{\alpha,\alpha'=1,\\\nu\nu' = 1}}^d \sum_{n, n' = 1}^{k} e^{i(n-n')(\mu_{\alpha} + \mu_{\nu} - \mu_{\alpha'} - \mu_{\nu'})}
 T_{\alpha\alpha'\nu\nu'}^{\beta\beta'}} & = \frac{1}{\binom{d+1}{2}} + \frac{(d-1)}{ d (k+1)(d+1)} \\
 & = \frac{1}{\binom{d+1}{2}} + O\qty(\frac{1}{k+1}),
\end{align}
where the last term vanishes as $k\rightarrow\infty$. The average value $\ev{\epsilon} = 2(d-1)/(k_1)$ is retrieved automatically.

\subsection{Reconstruction of states using \texorpdfstring{$\epsilon$}{}-2-designs} \label{app:random_reco}
Consider $\epsilon$-2-design $\qty{\ket{\psi_j}}$ as given in Definition~\ref{dfn:epsilon_des} and define
\begin{equation}
 \tilde{S} = \frac{1}{m}\sum_{j=1}^m \op{\psi_j}^{\otimes 2} = \underbrace{\frac{1}{d(d+1)}\sum_{\alpha\beta}\op{\alpha\beta}+\op{\alpha\beta}{\beta\alpha}}_{=S} + \Delta,
\end{equation}
where $S = \frac{2}{d(d+1)}\Pi_{\text{sym}}$ is the projection operator on the symmetric subspace and the correction term $\Delta$ has bounded $\infty$-norm, $\norm{\Delta}_{\infty} \leq \delta$. It is simple to show, following Scott~\cite{scott2006tight}, that
\begin{equation}
 \Tr_A\qty[\qty(\rho\otimes\mathbbm{1}) \tilde{S}] = \frac{1}{m}\sum_{j=1}^m p_j  \op{\psi_j} = \frac{\rho + \mathbbm{1}}{d(d+1)} + \Tr_A\qty[\qty(\rho\otimes\mathbbm{1}) \Delta],
\end{equation}
with $p_j = \ev{\rho}{\psi_j}$. From the above we can get an approximate reconstruction formula
\begin{equation}
 \tilde{\rho} = \rho + d(d+1)\Tr_A\qty[\qty(\rho\otimes\mathbbm{1}) \Delta] = \frac{d}{m} \sum_{j=1}^m (p_j (d+1) - 1) \op{\psi_j}.
\end{equation}

We may calculate the quality of this approximation using the infinity norm

\begin{equation}
 \norm{\tilde{\rho}-\rho}_\infty = d(d+1)\norm{\Tr_A\qty[\qty(\rho\otimes\mathbbm{1}) \Delta]}_\infty \leq d(d+1)\norm{\Delta}_\infty,
\end{equation}
where the last inequality follows from a simple chain of equalities,

\begin{align}
 \norm{\Tr_A\qty[\qty(\rho\otimes\mathbbm{1}) \Delta]}_\infty = \max_{\ket{\psi_A}\in\mathcal{H}_{A},\ket{\psi_B}\in\mathcal{H}_B}\abs{\ev{\Delta}{\psi_A\otimes\psi_B}} \leq \max_{\ket{\psi}\in\mathcal{H}_{AB}} \abs{\ev{\Delta}{\psi}} = \norm{\Delta}_\infty,
\end{align}
which finishes the proof.

 \section{Examples of bases producing simplex $2$-designs}
 \label{app:decoherence_examples}
 In this section, we show exemplary numerical bases $V_d$ for dimensions $d=5,6,7$ which produce simplex $2$-designs by decoherence:
\begin{itemize}
    \item Dimension $d=5$:
    {\scriptsize
    \begin{equation*}
    V_5=    \begin{pmatrix}
-0.204 - 0.197i & 0.316 + 0.243i & -0.156 + 0.376i & 0.017 - 0.278i & -0.309 - 0.650i \\
-0.282 + 0.666i & -0.327 + 0.036i & -0.239 - 0.277i & 0.267 + 0.103i & -0.291 - 0.261i \\
0.045 - 0.363i & -0.724 - 0.021i & 0.327 + 0.119i & 0.262 - 0.269i & -0.283 - 0.018i \\
0.331 - 0.045i & 0.258 - 0.205i & 0.207 - 0.690i & 0.068 - 0.422i & -0.105 - 0.259i \\
0.373 - 0.112i & -0.107 + 0.298i & 0.193 - 0.151i & -0.330 + 0.639i & -0.279 - 0.306i
\end{pmatrix}
    \end{equation*}
    }
The error, given by 1-normed based distance from a exact simplex $2$-design for this matrix is $\varepsilon_5=8.73\times10^{-13}$.
    \item Dimension $d=6$:
    {\scriptsize
    \begin{equation*}
        V_6=\begin{pmatrix}
-0.384 + 0.108i & 0.130 + 0.437i & -0.623 - 0.243i & -0.319 + 0.0929i & 0.0556 - 0.0702i & 0.187 - 0.179i \\
-0.299 - 0.595i & -0.0731 - 0.0112i & -0.256 + 0.160i & 0.374 + 0.278i & -0.312 - 0.263i & -0.0976 + 0.258i \\
-0.330 - 0.320i & 0.265 - 0.599i & -0.0245 - 0.133i & -0.0324 + 0.0956i & 0.260 + 0.358i & 0.0802 - 0.359i \\
0.207 + 0.105i & -0.128 - 0.143i & -0.369 - 0.313i & 0.136 - 0.176i & -0.546 + 0.366i & -0.403 - 0.178i \\
-0.168 - 0.315i & 0.0912 + 0.315i & 0.235 + 0.237i & -0.195 - 0.367i & -0.0180 - 0.128i & -0.503 - 0.460i \\
-0.0305 - 0.0636i & -0.414 - 0.205i & 0.0692 + 0.304i & -0.634 + 0.195i & -0.424 + 0.0369i & 0.228 - 0.106i
\end{pmatrix}
    \end{equation*}
    }
 With an error of $\varepsilon_6=9.87\times10^{-13}$.

\item Dimension $d=7$:
    {\scriptsize
    \begin{equation*}
        V_7=\begin{pmatrix}
-0.34 + 0.12i & 0.63 + 0.16i & 0.27 + 0.25i & -0.05 - 0.11i & -0.34 - 0.01i & -0.03 + 0.36i & -0.17 + 0.12i \\
0.15 + 0.15i & -0.07 - 0.21i & 0.36 + 0.27i & 0.19 - 0.17i & -0.16 - 0.23i & -0.43 - 0.48i & 0.01 + 0.39i \\
0.23 - 0.35i & 0.39 - 0.03i & -0.08 - 0.05i & -0.58 - 0.25i & -0.15 - 0.13i & 0.08 - 0.41i & 0.13 - 0.17i \\
0.15 - 0.14i & 0.31 + 0.08i & -0.01 - 0.16i & 0.31 - 0.21i & 0.22 + 0.61i & -0.29 - 0.19i & -0.35 - 0.13i \\
0.33 - 0.14i & 0.00 + 0.26i & 0.63 + 0.11i & -0.22 + 0.38i & 0.40 - 0.05i & 0.09 + 0.06i & -0.18 + 0.00i \\
-0.14 - 0.18i & -0.21 - 0.35i & 0.37 + 0.05i & 0.17 - 0.36i & -0.05 - 0.21i & 0.14 + 0.08i & -0.26 - 0.59i \\
-0.12 - 0.64i & -0.21 - 0.03i & 0.08 + 0.28i & 0.11 - 0.10i & -0.16 + 0.34i & 0.34 + 0.02i & 0.09 + 0.41i
\end{pmatrix}
    \end{equation*}
    }
In this case the error is $\varepsilon_7=9.63\times10^{-13}$.

 Note that the matrices presented here have been rounded, with the actual achievable precision on the order of $10^{-13}$. The error is calculated by the 1-norm of the difference of the averages obtained by this algorithm and the objective values. The full matrices and the code used for generating this solutions is available at \cite{GitHubTest}. 
\end{itemize}

\section{Examples of cyclic designs with \texorpdfstring{$k < d(d-1)$}{}}\label{app:num}



In this section, we will show the numerical solutions for $d=4$ that we have found but do not belong to the analytical solutions shown in previous sections, \textit{i.e.}, solutions that do not depend on the same difference sets with $\lambda=1$. First, we present solutions for which we have applied restrictions on the parameters $C_j$ in (\ref{m_h}), we obtained a unitary matrix and show its eigenvalues for certain number of bases.
\begin{itemize}
 \item $7$ bases, $k=6$
 {\scriptsize \begin{equation}
 U_7= \left(
 \begin{array}{cccc}
 0.1565 & 0.2616 & 0.3880 & 0.1939 \\
 -0.1560-0.1336 i & -0.1451-0.2438 i & 0.1643+0.2125 i & 0.1368+0.1649 i \\
 -0.0198+0.2456 i & 0.0573-0.2178 i & 0.0198-0.2456 i & -0.0573+0.2178 i \\
 0.1176-0.1236 i & -0.2055+0.1403 i & 0.2659-0.1841 i & -0.1780+0.1674 i \\
 \end{array}
 \right).
 \end{equation} }
 
 Dephased eigenvalues $\sigma_7=\left(\exp\left(\frac{2\pi i}{7}\times\{0, 1,3,4\}\right)\right).$

 \item $9$ bases, $k=8$
 {\scriptsize \begin{equation}
 U_9= \left(
 \begin{array}{cccc}
 0.1380 & 0.4553 & 0.2527 & 0.1540 \\
 -0.1886+0.0371 i & 0.4377-0.0542 i & -0.1522-0.1099 i & -0.0968+0.1269 i \\
 0.1386-0.1618 i & 0.0500+0.1124 i & -0.3348-0.0381 i & 0.1462+0.0875 i \\
 0.1854-0.0469 i & 0.1448-0.1319 i & -0.0552+0.1925 i & -0.2750-0.0137 i \\
 \end{array}
 \right).
 \end{equation}}
 
 Dephased eigenvalues $\sigma_9=\left(\exp\left(\frac{2\pi i}{9}\times\{0, 1,3,4\}\right)\right).$
 
 \item $11$ bases, $k=10$
 {\scriptsize\begin{equation}
 U_{11}= \left(
 \begin{array}{cccc}
 0.3249 & 0.2358 & 0.2121 & 0.2272 \\
 -0.2657-0.0847 i & 0.2392+0.0810 i & 0.2427+0.0288 i & -0.2162-0.0250 i \\
 0.1402-0.2193 i & -0.0998+0.2376 i & 0.0451-0.2352 i & -0.0855+0.2169 i \\
 -0.1391+0.2334 i & -0.0807+0.2086 i & 0.0807-0.2086 i & 0.1391-0.2334 i \\
 \end{array}
 \right).
 \end{equation}}
 
 Dephased eigenvalues $\sigma_{11}=\left(\exp\left(\frac{2\pi i}{11}\times\{0, 2,4,6\}\right)\right).$
\end{itemize}

We also Searched for solutions without any restrictions on parameters $C_j$. In the following we show these solutions:
\begin{itemize}
 \item $5$ bases, $k=4$ (Cyclic MUB)
 {\scriptsize\begin{equation}
 \tilde{U}_{5}= \left(
 \begin{array}{cccc}
 0.2500 & 0.2500 & 0.2500 & 0.2500 \\
 -0.2395+0.0718 i & -0.2395+0.0718 i & 0.2395-0.0718 i & 0.2395-0.0718 i \\
 0.2001+0.1499 i & -0.2001-0.1499 i & -0.2001-0.1499 i & 0.2001+0.1499 i \\
 0.2410+0.0664 i & -0.2410-0.0664 i & 0.2410+0.0664 i & -0.2410-0.0664 i \\
 \end{array}
 \right).
 \end{equation}}
 

 
 Dephased eigenvalues $\Tilde{\sigma}_5=\left(\exp\left(\frac{2\pi i}{5}\times\{0,2,3,4 \}\right)\right). $ 
 
 \item $7$ bases, $k=6$
 {\scriptsize\begin{equation}
 \tilde{U}_7= \left(
 \begin{array}{cccc}
 0.3918 & 0.1304 & 0.1894 & 0.2884 \\
 -0.0308-0.2224 i & 0.2120-0.1483 i & -0.0792+0.1593 i & -0.1019+0.2115 i \\
 -0.1537-0.2263 i & 0.0023+0.1476 i & 0.1534+0.2714 i & -0.0019-0.1927 i \\
 -0.3214-0.0984 i & 0.0988+0.1222 i & -0.0988-0.1222 i & 0.3214+0.0984 i \\
 \end{array}
 \right).
 \end{equation}}
 
 
 
 Dephased eigenvalues $\Tilde{\sigma}_7=\left(\exp\left(\frac{2\pi i}{7}\times\{0,1,2,3 \}\right)\right). $ 
 
 \item $8$ bases, $k=7$
 {\scriptsize \begin{equation}
 \tilde{U}_8=\left(
 \begin{array}{cccc}
 0.1953 & 0.2632 & 0.3040 & 0.2374 \\
 -0.1895-0.1245 i & 0.1895+0.1245 i & -0.2497-0.0990 i & 0.2497+0.0990 i \\
 -0.1005+0.2205 i & -0.0552+0.2456 i & 0.1165-0.2903 i & 0.0392-0.1759 i \\
 0.2028+0.0767 i & -0.2780-0.0435 i & -0.2013-0.0324 i & 0.2764-0.0008 i \\
 \end{array}
 \right).
 \end{equation}}
 
 
 
 Dephased eigenvalues $\Tilde{\sigma}_8=\left(\exp\left(\frac{2\pi i}{8}\times\{0,1,3,5 \}\right)\right). $ 
 
 \item $11$ bases, $k=10$
 {\scriptsize \begin{equation}
 \tilde{U}_{11}= \left(
 \begin{array}{cccc}
 0.0989 & 0.2591 & 0.3738 & 0.2681 \\
 0.1056+0.1481 i & 0.2093-0.2065 i & -0.3020-0.0793 i & -0.0129+0.1376 i \\
 0.1710+0.0568 i & -0.2501+0.1659 i & 0.0636-0.3213 i & 0.0154+0.0986 i \\
 0.0286-0.1509 i & 0.0849-0.0910 i & 0.1591-0.0628 i & -0.2727+0.3047 i \\
 \end{array}
 \right).
 \end{equation}}
 
 
 
 Dephased eigenvalues $\Tilde{\sigma}_{11}=\left(\exp\left(\frac{2\pi i}{11}\times\{0,3,4,6 \}\right)\right). $ 
 
 \item $12$ bases, $k=11$
 {\scriptsize \begin{equation}
 \tilde{U}_{12}= \left(
 \begin{array}{cccc}
 0.6423 & 0.0139 & 0.0272 & 0.3167 \\
 0.3076+0.0853 i & 0.0267+0.0366 i & -0.0770+0.0650 i & -0.2573-0.1870 i \\
 -0.2470-0.0170 i & 0.0544-0.0414 i & -0.0797+0.0521 i & 0.2723+0.0063 i \\
 0.0775-0.2460 i & -0.0832+0.0042 i & -0.0668+0.0526 i & 0.0725+0.1892 i \\
 \end{array}
 \right).
 \end{equation}}
 
 
 
 Dephased eigenvalues $\Tilde{\sigma}_{12}=\left(\exp\left(\frac{2\pi i}{12}\times\{0,1,5,10 \}\right)\right). $ 
\end{itemize}

All matrices are available in an online repository~\cite{data}
\end{widetext}

{\color{white}a}
\break
\bibliographystyle{quantum_abbr}
\bibliography{references}

\end{document}